\documentclass[11pt]{article}

\usepackage{amsmath, amssymb, amsthm, aligned-overset} 
\usepackage{thm-restate}
\usepackage[noend]{algorithm2e, algorithmic}
\usepackage{tikz}
\usepackage{float}
\usepackage{comment}
\usepackage{dirtytalk}
\usepackage[english]{babel}

\usepackage[margin=1in]{geometry}

\usepackage{amsmath}
\usepackage{graphicx}
\usepackage[colorlinks=true, allcolors=blue]{hyperref}
\usepackage{cleveref}

\theoremstyle{plain}
\newtheorem{theorem}{Theorem}
\newtheorem{lemma}{Lemma}[section]
\newtheorem{claim}[lemma]{Claim}
\newtheorem{fact}[lemma]{Fact}

\newtheorem{observation}[lemma]{Observation}

\newtheorem{remark}[lemma]{Remark}
\newtheorem{definition}[lemma]{Definition}

\DeclareMathOperator{\polylog}{polylog}
\DeclareMathOperator{\poly}{poly}
\newcommand{\emphdef}[1]{{\sf {#1}}}

\newcommand{\ts}{\texorpdfstring}

\newcommand{\cP}{\mathcal{P}}
\newcommand{\cD}{\mathcal{D}}

\newcommand{\cU}{\mathcal{U}}
\newcommand{\cL}{\mathcal{L}}

\newcommand{\eps}{\varepsilon}
\newcommand{\NP}{\mathsf{NP}}
\newcommand{\coRP}{\mathsf{coRP}}

\newcommand{\N}{\mathbb{N}}
\newcommand{\E}{\mathbb{E}}
\newcommand{\ra}{\rightarrow}

\newcommand{\cB}{\mathcal{B}}
\newcommand{\tO}{\widetilde{O}}

\newcommand{\cBnA}{\cB_{A}}
\newcommand{\defeq}{:=}

\title{Graph $k$-Coloring in Average Sublinear Time}

\author{Cassandra Marcussen\thanks{School of Engineering and Applied Sciences, Harvard University, Cambridge, Massachusetts, USA. Email: cmarcussen@g.harvard.edu. This material is based upon work supported by the Air Force Office of Scientific Research under award number FA9550-23-F-0014 (NDSEG Fellowship). Supported in part by NSF Award 2152413 and a Simons Investigator Award to Madhu Sudan.} \and Edward Pyne\thanks{Department of Electrical Engineering and Computer Science, MIT, Cambridge, Massachusetts, USA. Email: epyne@mit.edu. Supported by an NSF Graduate Research Fellowship.} \and Ronitt Rubinfeld\thanks{Computer Science and Artificial Intelligence Laboratory, MIT, Cambridge, Massachusetts, USA. Email: ronitt@csail.mit.edu. Supported by the NSF TRIPODS program (award DMS-2022448) and CCF-2310818.}  \and Asaf Shapira\thanks{School of Mathematics, Tel Aviv University, Tel Aviv 69978, Israel. Email: asafico@tau.ac.il. Supported in part by ERC Consolidator Grant 863438.} \and Shlomo Tauber \thanks{School of Computer Science, Tel Aviv University, Tel Aviv 69978, Israel. Email: shlomotauber@mail.tau.ac.il. Supported in part by ERC Consolidator Grant 863438.}}

\begin{document}
\begin{titlepage}
\date{}
\maketitle

\begin{abstract}
Graph $k$-coloring is one of the classic NP-complete problems. Previous work has studied its average time complexity, defined to be the average runtime of computing a $k$-coloring over the set of all $k$-colorable graphs on $n$ vertices. A highly influential result of Dyer-Frieze from 1989 gave an algorithm
with $O(n^2)$ average runtime for constant $k$.
This quadratic runtime appeared natural (and possibly even optimal)
since almost all $k$-colorable graphs have $\Theta(n^2)$ edges, so one
needs at least this time in order to read the (entire) input.
However, this was later improved by Ku{\v{c}}era in 1995 to average
runtime $O(n^2/k)$ for every $k \leq n^{c}$ where $c \in (0, 1)$.
Nevertheless, in the most interesting case of $k = O(1)$, the best-known
bound remained quadratic in $n$. The true average complexity of the
$k$-coloring problem has remained elusive for the last three decades.

We break the longstanding quadratic barrier. Our main result in this paper shows that the exact average-case complexity of this fundamental problem is $\Theta(nk)$
for every $k \leq n^{c'}$ and some $c' \in (0, 1)$. For $k = O(1)$, this reveals the average sublinear nature of $k$-colorability: the average-case complexity is \textit{linear in $n$}, and thus \textit{sublinear in the size of the input}. We further show that our $\Theta(nk)$ average runtime is \textit{optimal}, since a simple bound proves that every algorithm that correctly $k$-colors all $k$-colorable graphs requires $\Omega(n  k)$ average runtime.

Our proofs draw on ideas from sublinear and local algorithms and also yield a local computation algorithm (LCA) for $k$-coloring with average-case probe complexity $\poly(k)$. A key new ingredient in our algorithm is a method for certifying the unique colorability of random subgraphs, using tools from the theory of graph regularity.
\end{abstract}

\addtocounter{page}{-1}
\thispagestyle{empty}
\end{titlepage}

\section{Introduction}

While there are simple linear-time algorithms that can determine whether a graph is 2-colorable, the problem of determining whether a graph is $k$-colorable for $k \geq 3$ is one of the classic NP-hard problems~\cite{karp72}.  
Even approximating the chromatic number of a graph up to a multiplicative factor of $n^{1-\eps}$ is hard, assuming $\NP\neq\coRP$~\cite{FK}. 
Nevertheless,
there has been significant interest and success in finding algorithms for NP-hard problems whose 
\textit{average} runtime (over all inputs) is polynomial under various distributions
\cite{dyer1989solution, alon1994spectral, kuvcera1995expected, krivelevich2002deciding, DBLP:conf/random/ScottS03, DBLP:conf/soda/KrivelevichV06, DBLP:journals/corr/BourgeoisCDP15, DBLP:journals/rsa/AlonK20, DBLP:journals/corr/Anastos21}. For the specific case of coloring, a line of work beginning with the seminal paper of Dyer and Frieze~\cite{dyer1989solution} has shown the existence of coloring algorithms 
with polynomial average runtime over uniformly random $k$-colorable graphs.\footnote{Studying the uniform distribution over $k$-colorable graphs is a natural choice, adopted in many prior works (e.g., \cite{turner1988almost, dyer1989solution, kuvcera1993coloring, kuvcera1995expected, coja2010almost}). While one may initially think the most natural distribution is the uniform distribution over all graphs, as \cite{kuvcera1977expected} notes, a graph drawn from this distribution is extremely unlikely to be $k$-colorable and thus a trivial algorithm (that attempts to certify non-$k$-colorability) is already very efficient. We discuss related work in several other models in~\Cref{sec:related}.} The works of 
Dyer-Frieze~\cite{dyer1989solution} and Ku\v{c}era~\cite{kuvcera1995expected} achieve average runtime $O(n^2)$ and $O(n^2/k)$ respectively,\footnote{Dyer and Frieze \cite{dyer1989solution} state that their algorithm holds for constant $k$. The algorithm of~\cite{kuvcera1995expected} works for $k=o(\sqrt{n/
\log^2 n})$.} where $k$ is the chromatic number. In the well-studied setting of $k$ constant, both algorithms achieve $O(n^2)$ average runtime. This might seem optimal, since almost all $k$-colorable graphs have $\Theta(n^2)$ edges, so any deterministic algorithm must spend at least $\Theta(n^2)$ time just to read the input. Ku\v{c}era's $O(n^2/k)$ runtime revealed the sublinear possibilities of the $k$-coloring problem. Our main goal in this paper is to fully determine the average-case dependence on $n$ and $k$ of this fundamental problem.

Let us point out that there is also a line of work giving polynomial-time algorithms that $k$-color all but a small (polynomial) fraction of the $k$-colorable $n$-vertex graphs; see, for example, \cite{kuvcera1977expected, turner1988almost, blum1995coloring, coja2010almost}. We should emphasize that such results do not directly imply coloring algorithms with polynomial average runtime: an algorithm that uses the above to quickly color all 
but a $1/n^c$ fraction of  the $k$-colorable
graphs and resorts to an exponential time algorithm on the remaining
graphs would yield an exponential overall average running time. The problem of average-case $k$-coloring has become foundational in the study of the dichotomy between worst-case and average-case complexities. However, the exact complexity of $k$-coloring has remained unresolved. We approach this problem via a method inspired by sublinear and local algorithms.

\subsection{Our Results}\label{sec:our-results}

We give a tight characterization
of the average runtime of this fundamental NP-hard problem with respect to both $n$ \textit{and} $k$, for a large range of $k = k(n)$: we show there is an average $O(n k)$ randomized algorithm for $k$-coloring, and prove an $\Omega(nk)$ lower bound. Our upper bound breaks the longstanding quadratic barrier for $k$-coloring with constant $k$, giving the \textit{first subquadratic bound} for $k = O(1)$. For the lower bound, while an $\Omega(n)$ lower bound is immediate because the $k$-coloring being output has size $\Omega(n)$, establishing an $\Omega(nk)$ lower bound is more subtle.

In particular, when $k$ is constant, our result implies that the complexity of $k$-coloring is \textit{linear} in $n$, which is sublinear in the size of the graph. 
\begin{restatable}{theorem}{mainRand}\label{thm:mainRand}
    There is a randomized $k$-coloring algorithm with average runtime $O(n k)$
    over the uniform distribution over all $k$-colorable graphs, for $k \le n^c$ and $c = 1/37$.
\end{restatable} 

We prove \Cref{thm:mainRand} in Sections \ref{sec:goodProp} and \ref{sec:good-n-good-k}, and we also give a proof overview in \Cref{sec:overview}. This theorem holds in the access model where the graph is written in adjacency matrix format, and each query asking whether a pair of vertices is an edge takes constant time. 

A key new ingredient and conceptual contribution is a method for certifying that a small random subgraph is uniquely colorable, meaning the coloring is unique up to permuting the color classes. Our algorithm also crucially leverages techniques from local and sublinear algorithms. 
In fact, our algorithmic subroutines can be used to construct a Local Computation Algorithm (LCA)\footnote{In this context, a Local Computation Algorithm answers questions of the form \say{what is the color of vertex $v$,} where each query is answered with strongly sublinear work, and the provided answers are consistent between queries.} \cite{DBLP:conf/innovations/RubinfeldTVX11, DBLP:conf/soda/AlonRVX12, beyond_worst_case_LCA} for $k$-coloring that works well on average over
the distribution of input graphs. We prove there is an LCA that, with all but exponentially small probability, finds the color of 
any queried vertex; averaged over the distribution of graphs $G$, the maximum 
over $u \in G$ of the expected number of probes made is $\text{poly}(k)$, where the expectation is over the algorithm's random coins.\footnote{The $\text{poly}(k)$ comes from additive terms we need when implementing our algorithm, which are subsumed by the $O(nk)$ in the case of the global algorithm.} 
See \Cref{sec:LCA-implementation} for the LCA implementation.

We emphasize that an $o(n^2)$ average runtime, for $k = O(1)$, 
was not even known in the weaker setting where all but a polynomially
small fraction of $k$-colorable graphs must be colored. 
Additionally, our bound improves on Ku\v{c}era's \cite{kuvcera1995expected} bound for all $k$ for which our result holds.

At the cost of $\polylog(n)$ factors, we can obtain a deterministic algorithm:
\begin{restatable}{theorem}{mainDet}\label{thm:mainDet}
     There is a deterministic $k$-coloring algorithm with average runtime $\tO(n k)$ over the set of all $k$-colorable graphs, for $k \le n^c$ and $c = 1/37$. 
\end{restatable}

We prove \Cref{thm:mainDet} in \Cref{sec:deterministic} and give an overview of the proof in \Cref{sec:overview}.

In both the randomized and deterministic settings, our algorithms are sublinear in the input size, which is $\Theta(n^2)$. We show that this is \textit{optimal} by proving a runtime lower bound of $\Omega(nk)$:
\begin{restatable}{theorem}{lowerbound}\label{thm:lowerbound-intro}
    For every $k\le n^{1/3}$, any (randomized or deterministic) algorithm that correctly $k$-colors every $k$-colorable graph must have average runtime $\Omega(nk)$.
 \end{restatable}
In fact, we show that such an algorithm cannot be faster than this on \textit{any} input. We prove \Cref{thm:lowerbound-intro} in \Cref{sec:LB} and provide a proof overview in \Cref{sec:overview}.

\paragraph{Our upper-bound approach, at a high level:}
At an extremely high level, the algorithms of \Cref{thm:mainRand} and \Cref{thm:mainDet} sample random subgraphs (which we call the core), color this core, certify the coloring as unique, then attempt to propagate this coloring in several phases. In the event that many rounds of this high-level strategy fail, we resort to a brute-force algorithm. The certification of the core relies on ideas from algorithmic graph regularity. The propagation technique leverages ideas from sublinear algorithms. We give a (much) more detailed discussion of the proof in~\Cref{sec:alg-overview-coloring}.

\paragraph{Our lower-bound approach, at a high level:} We prove that any algorithm for $k$-coloring must make $\Omega(nk)$ queries on every graph, because with fewer queries many of the vertices remain insufficiently determined. This would allow an adversary to add an unqueried edge between vertices of the same color and force the algorithm's unchanged output to become invalid. See \Cref{sec:lb-intro-overview}.

\paragraph{Remark about alternative distributions.} While our results are stated for the uniform distribution over $k$-colorable graphs, they are robust to the underlying distribution over dense $k$-colorable graphs. We follow the outline of Dyer-Frieze \cite{dyer1989solution} by proving an $O(n k)$ runtime over a planted model where each color class is an approximately equal-sized set, and then transferring the result to the uniform distribution; see \Cref{sec:transfer-models}. Papers including Dyer-Frieze \cite{dyer1989solution} and Ku\v{c}era \cite{kuvcera1977expected} consider several other alternative distributions over dense $k$-colorable graphs, and our algorithm works in these models as well, by essentially the same reduction.

\subsection{Related Work}\label{sec:related}

Understanding the ability to $k$-color random instances of graphs --- as well as the ability to certify or refute $k$-colorability, extensions to semirandom models, and thresholds for $k$-colorability --- has remained a central goal in average-case complexity for the past four decades \cite{turner1988almost, dyer1989solution, kuvcera1977expected, kuvcera1993coloring,  alon1994spectral, blum1995coloring, kuvcera1995expected, promel1995random, DBLP:journals/rsa/AchlioptasF99, subramanian1999minimum, krivelevich2002coloring, krivelevich2002deciding, 
coja2004coloring, coja2004exact, bottcher2005coloring, DBLP:conf/analco/KrivelevichV06, DBLP:journals/jal/Coja-Oghlan07, DBLP:journals/ipl/BottcherV08, DBLP:conf/focs/AchlioptasC08, sommer2009note,  coja2010almost, DBLP:journals/combinatorics/Coja-Oghlan13, DBLP:conf/focs/Coja-OghlanV13, DBLP:conf/approx/BapstCHRV14, DBLP:journals/cpc/BapstCE17, DBLP:conf/colt/BandeiraBKMW21}. Many of these works (e.g., \cite{turner1988almost, kuvcera1993coloring})
studied the problem of efficiently $k$-coloring random graphs with high probability, though not high enough to imply average polynomial-time algorithms.\footnote{The extra challenge is that if the worst-case algorithm takes time $t(n)$, 
the probability of needing to run such an algorithm must be at most $\frac{1}{\Omega(t(n))}$. 
For $k$-coloring, where the worst-case algorithm takes time $\tO(2^n)$~\cite{coloringbruteforce}, 
the failure probability needs to be less than $2^{-n}$; even $2^{-n/2}$ does not suffice.} Dyer and Frieze's breakthrough work~\cite{dyer1989solution} was the first to study the average time complexity of $k$-coloring.

\paragraph{Comparing our algorithms to \cite{turner1988almost, kuvcera1993coloring, kuvcera1995expected}:}

We now compare our algorithms to prior works. First, in terms of techniques, one key difference between our work and previous work is our use of algorithmic {\em graph regularity} to certify the unique colorability of a randomly sampled subgraph. This, in turn, provides a sublinear-time certificate of the unique colorability of a subgraph on all-but-$O_k(1)$ vertices of a random $k$-colorable graph.

\paragraph{Average-Case Two-Coloring of Random Hypergraphs:} 
A recent paper \cite{hypergraphcoloring} studies 
the average-case complexity of $2$-coloring $\ell$-uniform hypergraphs. 
Our work shares a high-level algorithmic principle proposed by \cite{hypergraphcoloring}: 
The algorithms of \cite{hypergraphcoloring} and \Cref{thm:mainRand,thm:mainDet} both find a (small, uniquely colorable) subgraph and use it to
color a constant fraction of vertices. Then, a small random sample
of the newly colored vertices can be used to color all of the remaining vertices.
However,  there are additional challenges in the $k$-coloring of graphs.
The first challenge in generalizing from $2$-coloring to $k$-coloring is that more information is necessary to determine the color. Coupled with the challenge that average $k$-colorable 
graphs are much  sparser and have weaker connectivity properties than 
average $k$-colorable hypergraphs, this forces the addition of new phases handling small numbers of badly-connected vertices. 
More crucially, our algorithm 
achieves a linear dependence in $k$, whereas a straightforward adaptation of the approach 
in \cite{hypergraphcoloring} gives an exponential dependence. Their approach finds 
a small \textit{complete} subgraph to use as the core, for which verifying unique colorability is trivial. 
Since the density of these subgraphs is inverse exponential in $k$, they  require $2^{\Omega(k)}$ samples (and runtime) merely to identify a core.

\paragraph{Turner's algorithm:}
Turner's propagation-based algorithm runs in $O(n^2\log{n})$ time and
correctly colors at least $1- 1/\poly{n}$ fraction of $k$-colorable graphs (and therefore
does not yield a polynomial average time algorithm). The algorithm
initially identifies and colors a
$k$-clique in the input graph,  and then proceeds to 
repeatedly color vertices for which only one available color remains. 
This does not proceed in a fixed number of stages, 
which makes it unsuitable for a local sublinear time implementation (see \Cref{sec:LCA-implementation} for an implementation of our algorithm).

\paragraph{Ku\v{c}era's algorithms:}
Ku\v{c}era \cite{kuvcera1995expected} analyzes the expected time for $k$-coloring graphs, by 
considering the local statistic of shared neighbor counts to determine, pairwise, 
whether two vertices are in the same or different color classes. His algorithm achieves a running time of $O(n^2/k)$ for $k=o(\sqrt{n/
\log^2 n})$, with some additional preprocessing time. \cite{kuvcera1993coloring} gives a parallel (CRCW PRAM) algorithm for constant average parallel time $k$-coloring. 
Ku\v{c}era's algorithm \cite{kuvcera1993coloring} holds for $k = \log^{O(1)} n$, and does not immediately
yield a sequential runtime less than $\Theta(n^2)$. 
Due to the concurrent read and concurrent write ability in the
PRAM model, the Parnas-Ron technique \cite{DBLP:journals/tcs/ParnasR07} that converts constant-round
LOCAL distributed algorithms into efficient 
LCAs cannot be used here. 
Both papers derandomize their algorithms by using a component of the graph as a source of random bits. 
Because we have edge dependencies when conditioning on the random graph being $k$-colorable, we use a simpler approach that directly uses the vertex-symmetry of the distribution (see \Cref{sec:int:derand} for details). 

\paragraph{Sublinear-time coloring:}
Graph $(\Delta+1)$-coloring (where $\Delta$ is the maximum degree of the graph) has been considered in the context of local computation algorithms (LCAs) \cite{DBLP:conf/innovations/RubinfeldTVX11, DBLP:conf/podc/ChangFGUZ19, DBLP:journals/corr/abs-2103-10990, DBLP:conf/icalp/DorobiszK23} and other sublinear access models \cite{DBLP:journals/siamcomp/Linial92, DBLP:journals/jacm/HarrisSS18, DBLP:conf/stacs/BamasE19, DBLP:journals/topc/Maus23, DBLP:conf/soda/FlinGHKN24, DBLP:conf/soda/AssadiCK19, DBLP:conf/approx/AlonA20, DBLP:journals/corr/abs-2502-06024, DBLP:journals/theoretics/AssadiY26}. We are not aware of LCAs that achieve the optimal coloring for $k \geq 3$.

\subsection{Open Problems}
We view the most interesting open problem as extending our results to hold for the entire range of $k$. 
The proof of the key technical theorem (\Cref{thm:good_property_probability}) relies on $k \leq n^c$ for certain $c \in (0, 1)$ in several key places. For concreteness, the current limiting factor is certifying the unique colorability of a size-$\nu$ core; the degree and codegree conditions we certify require $\nu = \Omega( k^9 \log k)$. We then apply McDiarmid's inequality to say that the graph has many good cores with probability at least $1 - k^{-2n}$. This requires $k \leq n^{1/36 - \delta}$, where $36$ comes from $(\text{exponent in core size}) \cdot 4$.
Extending to the full range of $1 \leq k = k(n) \leq n$ will likely require a different algorithmic and analytic approach. 

Additionally, it remains a compelling goal to achieve a deterministic algorithm with average runtime $O(nk)$, i.e., without $\polylog(n)$ factors.

We are interested in whether the algorithmic principles we utilize can lead to new insights for the average-case complexity of other NP-hard problems, such as the other problems considered by Dyer and Frieze \cite{dyer1989solution} and Ku\v{c}era \cite{kuvcera1995expected}. Generally, as explained in \Cref{sec:our-results}, our algorithm relies on the principle of first coloring a (small random) subgraph and then propagating this structure to the rest of the graph efficiently. We wonder if this local propagation technique can be applied to achieve a sublinear average-case complexity for other NP-hard problems.

\subsection{Roadmap}
In~\Cref{sec:overview} we sketch the upper bound proofs. In~\Cref{sec:prelims} we recall prior results on coloring that we use as subroutines. In~\Cref{sec:goodProp} we prove most $k$-colorable graphs have desirable properties that we discuss later. In~\Cref{sec:good-n-good-k} we prove the main result. Finally, in~\Cref{sec:LB} we prove the $\Omega(nk)$ lower bound. In \Cref{sec:LCA-implementation}, we give an implementation of the local computation algorithm.

\section{Algorithm Overview and Proof Ideas}\label{sec:overview}

We first describe the $k$-coloring algorithm of \Cref{thm:mainRand} in more detail, and then the structure 
possessed by most $k$-colorable graphs that allows the algorithm to succeed.
Finally, we show how to derandomize and obtain the deterministic algorithm of \Cref{thm:mainDet}.

\subsection{The Coloring Algorithm}\label{sec:alg-overview-coloring}

Our algorithm consists of two main stages, described in \Cref{sec:first-stage} and \Cref{sec:second-stage}. 
Each iteration of the first stage runs in time $T_1=O(kn+\poly(k))$ and either produces a valid coloring or returns \textsc{FAIL}. We show that for a $1-q_1$ fraction of $k$-colorable graphs (which we call \say{awesome} graphs), each iteration of the first stage produces a valid coloring with probability $1/2$, and hence we perform a constant number of iterations in expectation and color awesome graphs in expected time $O(T_1)$. 

Unfortunately, since the fraction of non-awesome graphs is $q_1\gg 2^{-n}$, we cannot simply brute-force all non-awesome graphs. Thus, we introduce a second stage, which runs in time $T_2=n^{\poly(k)}$. We prove that it colors a $1-q_2$ fraction of $k$-colorable graphs (which we call \say{okay} graphs), where $q_2\ll q_1$. If even the second stage does not color a graph, the algorithm will exhaustively search for a coloring in time $T_{BF} = \tO(2^n)$, which can be tolerated because $q_2$ is extremely small.
Thus, our final expected runtime is
\begin{align*}
    \E_{G}[T(G)] &\le \E_{G:\text{$G$ awesome}}[T(G)] + q_1\cdot \E_{G:\text{$G$ not awesome}}[T(G)]\\
    &\le \E_{G:\text{$G$ awesome}}[T(G)] + q_1\cdot \E_{G:\text{$G$ okay}}[T(G)] + q_2\cdot \E_{G:\text{$G$ not okay}}[T(G)]\\
    &\le O(T_1) + q_1\cdot T_2+q_2 \cdot T_{BF} = O(T_1).
\end{align*}

\subsubsection{The First-Stage Algorithm}\label{sec:first-stage}

Suppose we are given a graph $G$ that we must color. We first describe a single iteration of Stage 1 (and recall that we will repeat this up to $m=\poly(n)$ times if we do not find a coloring).
\begin{enumerate}
    \item (Phase 1) We first attempt to find a \say{good core,} which we can color and then extend this partial coloring outward. 
    We randomly sample a set of vertices $D$ of size $\nu=\text{poly}(k)$, which we call the core. Our coloring algorithm attempts to color the core and propagate this coloring outwards. However, in order for the algorithm to never output an incorrect coloring, we must certify that the core is \textit{uniquely colorable}. We say a graph is uniquely $k$-colorable if it has exactly one proper $k$-coloring (up to permuting the names of the $k$ color classes). If $D$ were not uniquely colorable, it is possible the algorithm may create an incorrect coloring in subsequent phases. Since our algorithm runs in sublinear time and cannot check all of the edges, this may go undetected in our average time bound.
    
    We prove that unique colorability is efficiently certifiable, meaning there is an efficiently checkable condition that implies unique colorability and is satisfied for most uniquely colorable graphs. The certificate consists of simple, locally checkable statistics. We verify that the coloring produced on the candidate good core has approximately balanced color classes, that every vertex has approximately the expected degree into each other class, and that every pair of vertices has approximately the expected codegree (number of common neighbors) into each other class (see \Cref{def:goodCore}). These conditions can be checked in $\poly(k)$ time, yet we prove that they are sufficient to force the coloring to be unique (\Cref{lem:core-unique-color}). Degree and codegree regularity imply $\varepsilon$-regularity between the color classes \cite{DBLP:journals/jal/AlonDLRY94}, and furthermore, $\varepsilon$-regularity rules out any alternative colorings, since a second coloring would need to place a large, edge-free pair of subsets across two different color classes of the first coloring, which is forbidden by regularity. While the $\varepsilon$-regularity implication of degree and codegree regularity draws from prior work, using these statistics as a $\poly(k)$-time,
    $\poly(k)$-size \textit{certificate of unique colorability} in the setting of average-case coloring is, to our knowledge, new. The certification of the core is essential as it is what enables the algorithm of \Cref{thm:mainRand} to be worst-case correct, since propagating from a core that is colorable in more than one way could produce an improper coloring.

    If the certification algorithm that verifies simple, locally checkable statistics fails to certify unique colorability, abort. Otherwise, we attempt to color this core using Ku\v{c}era's algorithm  \cite{kuvcera1995expected} that is fast on average (the choice of $\nu$ is made so that Ku\v{c}era's algorithm can be applied and the core is good and uniquely colorable). If Ku\v{c}era's algorithm runs for too long, we abort. Otherwise, we have that $D$ has a unique coloring that we can find and certify in time $\poly(k)$; we say that $D$ is a good core in this case. Let $D_i$ be the vertices in $D$ within color class $i$. If each $|D_i| \geq \nu/2k$, proceed to the next stage; else, repeat Phase 1. 
    
    \item (Phase 2)
    Next, suppose we have a good core $D$. We initialize sets $S_i$, $i \in [k]$. $S_i$ will consist of vertices that are colored $i$.
    The next phase of the algorithm will test if each vertex $v \in G \setminus D$ should be colored $i$ as follows. For each $v$, the algorithm will, for $O(k)$ steps, sample a random color $j$ (that has not yet been excluded) and a random $u\in D_j$, and if $(v,u)$ is in the graph we know $v$ cannot be colored with $j$, so we exclude it. If we exclude all but one color $i$ in this fashion, we add $v$ to $S_i$. It is easy to see that since $D$ is uniquely $k$-colorable, every element of $S_i$ must receive color $i$ in any legal coloring of $G$. Moreover, any $v$ that has many edges to all but one color class $D_i$ of the good core $D$ is colored by this procedure with probability $0.99$. By construction, the total runtime of this step is $O(n \cdot k)$.

    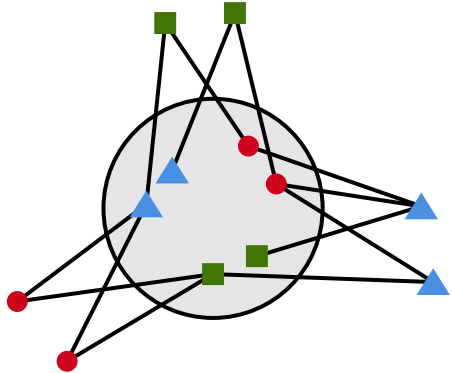
\begin{figure}[h]
        \centering
        \tikzset{every picture/.style={line width=0.75pt}} 

\begin{tikzpicture}[x=0.75pt,y=0.75pt,yscale=-1,xscale=1]

\draw  [fill={rgb, 255:red, 155; green, 155; blue, 155 }  ,fill opacity=0.26 ][line width=1.5]  (280,109) .. controls (280,78.62) and (304.62,54) .. (335,54) .. controls (365.38,54) and (390,78.62) .. (390,109) .. controls (390,139.38) and (365.38,164) .. (335,164) .. controls (304.62,164) and (280,139.38) .. (280,109) -- cycle ;
\draw [line width=1.5]    (236.75,155.75) -- (301.5,107) ;
\draw [line width=1.5]    (236.75,155.75) -- (335,142) ;
\draw [line width=1.5]    (261.75,185.75) -- (301.5,107) ;
\draw [line width=1.5]    (261.75,185.75) -- (335,142) ;
\draw [line width=1.5]    (352.75,77.75) -- (439.5,108) ;
\draw [line width=1.5]    (366.75,96.75) -- (439.5,108) ;
\draw [line width=1.5]    (357,133) -- (439.5,108) ;
\draw [line width=1.5]    (335,142) -- (445.5,146) ;
\draw [line width=1.5]    (366.75,96.75) -- (445.5,146) ;
\draw [line width=1.5]    (346,11) -- (366.75,96.75) ;
\draw [line width=1.5]    (346,11) -- (314.5,90) ;
\draw [line width=1.5]    (311,16) -- (352.75,77.75) ;
\draw [line width=1.5]    (311,16) -- (301.5,107) ;
\draw  [color={rgb, 255:red, 208; green, 2; blue, 27 }  ,draw opacity=1 ][fill={rgb, 255:red, 208; green, 2; blue, 27 }  ,fill opacity=1 ] (362,96.75) .. controls (362,94.13) and (364.13,92) .. (366.75,92) .. controls (369.37,92) and (371.5,94.13) .. (371.5,96.75) .. controls (371.5,99.37) and (369.37,101.5) .. (366.75,101.5) .. controls (364.13,101.5) and (362,99.37) .. (362,96.75) -- cycle ;
\draw  [color={rgb, 255:red, 208; green, 2; blue, 27 }  ,draw opacity=1 ][fill={rgb, 255:red, 208; green, 2; blue, 27 }  ,fill opacity=1 ] (348,77.75) .. controls (348,75.13) and (350.13,73) .. (352.75,73) .. controls (355.37,73) and (357.5,75.13) .. (357.5,77.75) .. controls (357.5,80.37) and (355.37,82.5) .. (352.75,82.5) .. controls (350.13,82.5) and (348,80.37) .. (348,77.75) -- cycle ;
\draw  [color={rgb, 255:red, 65; green, 117; blue, 5 }  ,draw opacity=1 ][fill={rgb, 255:red, 65; green, 117; blue, 5 }  ,fill opacity=1 ] (330,137) -- (340,137) -- (340,147) -- (330,147) -- cycle ;
\draw  [color={rgb, 255:red, 65; green, 117; blue, 5 }  ,draw opacity=1 ][fill={rgb, 255:red, 65; green, 117; blue, 5 }  ,fill opacity=1 ] (352,128) -- (362,128) -- (362,138) -- (352,138) -- cycle ;
\draw  [color={rgb, 255:red, 74; green, 144; blue, 226 }  ,draw opacity=1 ][fill={rgb, 255:red, 74; green, 144; blue, 226 }  ,fill opacity=1 ] (301.5,101) -- (309,113) -- (294,113) -- cycle ;
\draw  [color={rgb, 255:red, 74; green, 144; blue, 226 }  ,draw opacity=1 ][fill={rgb, 255:red, 74; green, 144; blue, 226 }  ,fill opacity=1 ] (314.5,84) -- (322,96) -- (307,96) -- cycle ;
\draw  [color={rgb, 255:red, 208; green, 2; blue, 27 }  ,draw opacity=1 ][fill={rgb, 255:red, 208; green, 2; blue, 27 }  ,fill opacity=1 ] (232,155.75) .. controls (232,153.13) and (234.13,151) .. (236.75,151) .. controls (239.37,151) and (241.5,153.13) .. (241.5,155.75) .. controls (241.5,158.37) and (239.37,160.5) .. (236.75,160.5) .. controls (234.13,160.5) and (232,158.37) .. (232,155.75) -- cycle ;
\draw  [color={rgb, 255:red, 65; green, 117; blue, 5 }  ,draw opacity=1 ][fill={rgb, 255:red, 65; green, 117; blue, 5 }  ,fill opacity=1 ] (306,11) -- (316,11) -- (316,21) -- (306,21) -- cycle ;
\draw  [color={rgb, 255:red, 65; green, 117; blue, 5 }  ,draw opacity=1 ][fill={rgb, 255:red, 65; green, 117; blue, 5 }  ,fill opacity=1 ] (341,6) -- (351,6) -- (351,16) -- (341,16) -- cycle ;
\draw  [color={rgb, 255:red, 74; green, 144; blue, 226 }  ,draw opacity=1 ][fill={rgb, 255:red, 74; green, 144; blue, 226 }  ,fill opacity=1 ] (439.5,102) -- (447,114) -- (432,114) -- cycle ;
\draw  [color={rgb, 255:red, 74; green, 144; blue, 226 }  ,draw opacity=1 ][fill={rgb, 255:red, 74; green, 144; blue, 226 }  ,fill opacity=1 ] (445.5,140) -- (453,152) -- (438,152) -- cycle ;
\draw  [color={rgb, 255:red, 208; green, 2; blue, 27 }  ,draw opacity=1 ][fill={rgb, 255:red, 208; green, 2; blue, 27 }  ,fill opacity=1 ] (257,185.75) .. controls (257,183.13) and (259.13,181) .. (261.75,181) .. controls (264.37,181) and (266.5,183.13) .. (266.5,185.75) .. controls (266.5,188.37) and (264.37,190.5) .. (261.75,190.5) .. controls (259.13,190.5) and (257,188.37) .. (257,185.75) -- cycle ;

\end{tikzpicture}
        \caption{Consider $k$-coloring for $k = 3$. Let the shaded circle be a uniquely colored good core. In Phase 2, for each vertex $v$ not in the core, the algorithm samples vertices in the core to find adjacencies to vertices in all but one color class. If such adjacencies are found, $v$ is then colored the remaining color.}
        \label{fig:propagation-1}
    \end{figure}

    \item (Phase 3) Finally, we enumerate over all remaining uncolored vertices. For each uncolored vertex $v$, we sample random vertices $u_j\in S_j$ and test if $(v,u_j)\in E$. If we find such vertices $u_j$ for all but one color $i$, we again add $v$ to color $i$. We continue in this fashion for $O(n \cdot k)$ total steps, and if we have colored every vertex, we declare success.

    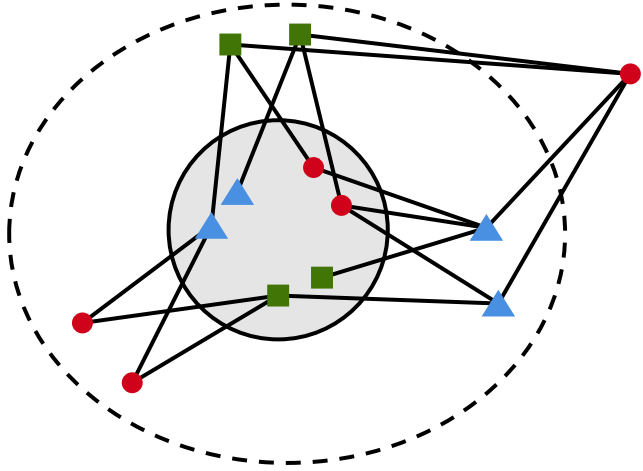
\begin{figure}[h]
        \centering
        \tikzset{every picture/.style={line width=0.75pt}} 

\begin{tikzpicture}[x=0.75pt,y=0.75pt,yscale=-1,xscale=1]

\draw [line width=1.5]    (328,25) -- (493.75,44.75) ;
\draw [line width=1.5]    (421.5,122) -- (493.75,44.75) ;
\draw [line width=1.5]    (427.5,160) -- (493.75,44.75) ;
\draw [line width=1.5]    (293,30) -- (493.75,44.75) ;
\draw  [fill={rgb, 255:red, 155; green, 155; blue, 155 }  ,fill opacity=0.26 ][line width=1.5]  (262,123) .. controls (262,92.62) and (286.62,68) .. (317,68) .. controls (347.38,68) and (372,92.62) .. (372,123) .. controls (372,153.38) and (347.38,178) .. (317,178) .. controls (286.62,178) and (262,153.38) .. (262,123) -- cycle ;
\draw [line width=1.5]    (218.75,169.75) -- (283.5,121) ;
\draw [line width=1.5]    (218.75,169.75) -- (317,156) ;
\draw [line width=1.5]    (243.75,199.75) -- (283.5,121) ;
\draw [line width=1.5]    (243.75,199.75) -- (317,156) ;
\draw [line width=1.5]    (334.75,91.75) -- (421.5,122) ;
\draw [line width=1.5]    (348.75,110.75) -- (421.5,122) ;
\draw [line width=1.5]    (339,147) -- (421.5,122) ;
\draw [line width=1.5]    (317,156) -- (427.5,160) ;
\draw [line width=1.5]    (348.75,110.75) -- (427.5,160) ;
\draw [line width=1.5]    (328,25) -- (348.75,110.75) ;
\draw [line width=1.5]    (328,25) -- (296.5,104) ;
\draw [line width=1.5]    (293,30) -- (334.75,91.75) ;
\draw [line width=1.5]    (293,30) -- (283.5,121) ;
\draw  [color={rgb, 255:red, 208; green, 2; blue, 27 }  ,draw opacity=1 ][fill={rgb, 255:red, 208; green, 2; blue, 27 }  ,fill opacity=1 ] (344,110.75) .. controls (344,108.13) and (346.13,106) .. (348.75,106) .. controls (351.37,106) and (353.5,108.13) .. (353.5,110.75) .. controls (353.5,113.37) and (351.37,115.5) .. (348.75,115.5) .. controls (346.13,115.5) and (344,113.37) .. (344,110.75) -- cycle ;
\draw  [color={rgb, 255:red, 208; green, 2; blue, 27 }  ,draw opacity=1 ][fill={rgb, 255:red, 208; green, 2; blue, 27 }  ,fill opacity=1 ] (330,91.75) .. controls (330,89.13) and (332.13,87) .. (334.75,87) .. controls (337.37,87) and (339.5,89.13) .. (339.5,91.75) .. controls (339.5,94.37) and (337.37,96.5) .. (334.75,96.5) .. controls (332.13,96.5) and (330,94.37) .. (330,91.75) -- cycle ;
\draw  [color={rgb, 255:red, 65; green, 117; blue, 5 }  ,draw opacity=1 ][fill={rgb, 255:red, 65; green, 117; blue, 5 }  ,fill opacity=1 ] (312,151) -- (322,151) -- (322,161) -- (312,161) -- cycle ;
\draw  [color={rgb, 255:red, 65; green, 117; blue, 5 }  ,draw opacity=1 ][fill={rgb, 255:red, 65; green, 117; blue, 5 }  ,fill opacity=1 ] (334,142) -- (344,142) -- (344,152) -- (334,152) -- cycle ;
\draw  [color={rgb, 255:red, 74; green, 144; blue, 226 }  ,draw opacity=1 ][fill={rgb, 255:red, 74; green, 144; blue, 226 }  ,fill opacity=1 ] (283.5,115) -- (291,127) -- (276,127) -- cycle ;
\draw  [color={rgb, 255:red, 74; green, 144; blue, 226 }  ,draw opacity=1 ][fill={rgb, 255:red, 74; green, 144; blue, 226 }  ,fill opacity=1 ] (296.5,98) -- (304,110) -- (289,110) -- cycle ;
\draw  [color={rgb, 255:red, 208; green, 2; blue, 27 }  ,draw opacity=1 ][fill={rgb, 255:red, 208; green, 2; blue, 27 }  ,fill opacity=1 ] (214,169.75) .. controls (214,167.13) and (216.13,165) .. (218.75,165) .. controls (221.37,165) and (223.5,167.13) .. (223.5,169.75) .. controls (223.5,172.37) and (221.37,174.5) .. (218.75,174.5) .. controls (216.13,174.5) and (214,172.37) .. (214,169.75) -- cycle ;
\draw  [color={rgb, 255:red, 65; green, 117; blue, 5 }  ,draw opacity=1 ][fill={rgb, 255:red, 65; green, 117; blue, 5 }  ,fill opacity=1 ] (288,25) -- (298,25) -- (298,35) -- (288,35) -- cycle ;
\draw  [color={rgb, 255:red, 65; green, 117; blue, 5 }  ,draw opacity=1 ][fill={rgb, 255:red, 65; green, 117; blue, 5 }  ,fill opacity=1 ] (323,20) -- (333,20) -- (333,30) -- (323,30) -- cycle ;
\draw  [color={rgb, 255:red, 74; green, 144; blue, 226 }  ,draw opacity=1 ][fill={rgb, 255:red, 74; green, 144; blue, 226 }  ,fill opacity=1 ] (421.5,116) -- (429,128) -- (414,128) -- cycle ;
\draw  [color={rgb, 255:red, 74; green, 144; blue, 226 }  ,draw opacity=1 ][fill={rgb, 255:red, 74; green, 144; blue, 226 }  ,fill opacity=1 ] (427.5,154) -- (435,166) -- (420,166) -- cycle ;
\draw  [color={rgb, 255:red, 208; green, 2; blue, 27 }  ,draw opacity=1 ][fill={rgb, 255:red, 208; green, 2; blue, 27 }  ,fill opacity=1 ] (239,199.75) .. controls (239,197.13) and (241.13,195) .. (243.75,195) .. controls (246.37,195) and (248.5,197.13) .. (248.5,199.75) .. controls (248.5,202.37) and (246.37,204.5) .. (243.75,204.5) .. controls (241.13,204.5) and (239,202.37) .. (239,199.75) -- cycle ;
\draw  [dash pattern={on 5.63pt off 4.5pt}][line width=1.5]  (182,125) .. controls (182,61.49) and (244.46,10) .. (321.5,10) .. controls (398.54,10) and (461,61.49) .. (461,125) .. controls (461,188.51) and (398.54,240) .. (321.5,240) .. controls (244.46,240) and (182,188.51) .. (182,125) -- cycle ;
\draw  [color={rgb, 255:red, 208; green, 2; blue, 27 }  ,draw opacity=1 ][fill={rgb, 255:red, 208; green, 2; blue, 27 }  ,fill opacity=1 ] (489,44.75) .. controls (489,42.13) and (491.13,40) .. (493.75,40) .. controls (496.37,40) and (498.5,42.13) .. (498.5,44.75) .. controls (498.5,47.37) and (496.37,49.5) .. (493.75,49.5) .. controls (491.13,49.5) and (489,47.37) .. (489,44.75) -- cycle ;

\end{tikzpicture}
        \caption{Consider $k$-coloring for $k = 3$. In Phase 3, for each vertex $v$ not colored in an earlier phase, vertices in the good core and those colored in Phase 2 are sampled. If vertices of every color but one are found, $v$ is colored with the remaining color.}
        \label{fig:propagation-2}
    \end{figure}
\end{enumerate}
Note that the above procedure runs in time $T_1=O(n \cdot k)$ as claimed.
We call a graph $G$ \emphdef{awesome} if, informally, this stage is likely to color it. More specifically, an awesome graph must have many good cores, and all of these cores must color all other vertices via the two-stage propagation algorithm described above. We give a formal definition of awesome graphs in~\Cref{sec:int:prop}. 

\subsubsection{Remaining Stages}\label{sec:second-stage}
We handle the non-awesome graphs as follows. We keep the basic structure of Stage 1, but modify it in two main ways:
\begin{enumerate}
    \item We now enumerate over \textit{all} cores of size $\nu = \text{poly}(k)$, for a larger core size, and attempt to produce a certified unique coloring for each.
    \item After attempting to color all vertices by looking for neighbors in $S_i$ in Phase 3, if there are $B=\tO(k)$ extra vertices that we have not yet assigned a color, we exhaustively search for a coloring on these remaining vertices that is consistent with the partial coloring created so far. This takes time $k^B \cdot B \cdot n$ and we call this Phase 4.

    \begin{figure}[h]
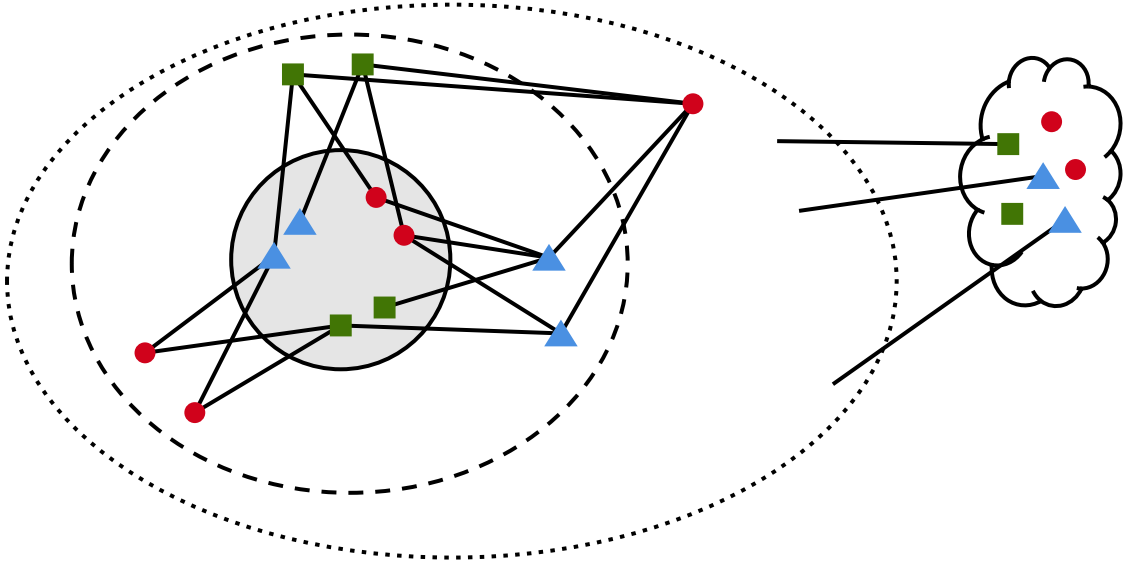

        \centering
        \include{Visuals/propagation-3}
        \caption{Consider $k$-coloring for $k = 3$. If $B = \widetilde{O}(k)$ extra vertices are not yet assigned a color, exhaustively search for the lexicographically first $k$-coloring of these vertices that is consistent with the partial coloring constructed so far.}
        \label{fig:propagation-3}
    \end{figure}
\end{enumerate}
Note that this modified procedure runs in time $T_2=\poly(\binom{n}{\nu}+k^B \cdot B \cdot n)$ 
as claimed. We again call a graph \emphdef{okay} if, essentially, this stage colors it. The key difference from awesome graphs is that we allow a small number of exceptional vertices that are not colored by propagation (and we choose a slightly larger core size). 

We prove that all but a $k^{-2n}$ fraction of graphs are okay. Thus, if $G$ is not colored in stage 2, the algorithm can tolerate exhaustively searching for the lexicographically first coloring.

\begin{remark}\label{rem:remaining-vertices}

    The small number of vertices that fail to be colored in Phase 2 is \textbf{not} an artifact of our proof. Fixing a vertex $u\in A_i$ and considering the edges to $S_j\subseteq A_j$ for some $j\neq i$ as sampled independently with probability $p=1/2$, the probability that $u$ has \textbf{no} edges to $S_j$ is approximately $2^{-|S_j|}\approx 2^{-n/k}$ (and thus certainly does not have a positive fraction of edges to $S_j$, and hence is not inferred by $S$). Since we wish to prove that our fast algorithm succeeds with probability at least $1-2^{-n}\gg 1-2^{-n/k}$ (so that we can afford to brute force on the bad cases), the \say{fast} algorithm must be able to deal with some vertices failing to be inferred by $S$. Luckily, the probability that more than $B=\tO(k)$ vertices fail in this way simultaneously is approximately $(2^{-n/k})^B\ll 2^{-n}$. Thus with all but exponentially small probability in $n$, the number of non-connected vertices is bounded by $O_k(1)$.
\end{remark}

The formal algorithmic pseudocode for each of the phases is given in \Cref{sec:algorithmic-subroutines}. The final $k$-coloring algorithm that combines these phases is given in \Cref{alg:main-rand}.

\subsection{The Graph Structure}\label{sec:int:prop} 
\newcommand{\cUnk}{\cU}
In order to argue that the above algorithm obtains a fast runtime, we must show that many graphs are awesome (i.e., colored in the first stage), and almost all are okay (i.e., colored in the second stage). To do so, we will first define several properties (motivated by the algorithm) that we hope graphs have, then prove that the desired fraction of graphs have them.

\subsubsection{Good cores and inferred colors}\label{sec:int:cert}

To begin with, let \textsc{Kucera1995} be the algorithm referenced in \Cref{thm:kuvcera}, which is a deterministic $k$-coloring algorithm of Ku\v{c}era \cite{kuvcera1995expected} that has average runtime $O(n^2/k)$.\footnote{\textsc{Kucera1995} works in all settings of $k$ that our algorithm is claimed for.} Recall that our algorithm begins by looking for a ``good core'' in a graph, defined as follows.  A good core $D$ is a subgraph that has three ``nice'' properties. 
First, the color classes of the core are ``approximately balanced,'' meaning roughly size $|D|/k$. 
Second,  \textsc{Kucera1995} produces a proper $k$-coloring in expected time $O(|D|^2/k)$. 
Third, the core satisfies simple conditions (approximate regularity of degrees and shared neighbor counts) that imply this coloring is unique. 
This last point is essential: while it is easy to show that most cores are uniquely colorable, 
we need to ensure that we \textit{know} when we do not have a uniquely colorable core, since a false positive could cause a correctness error.
The certification of unique colorability via degrees and shared neighbor counts relies on ideas from algorithmic graph regularity and quality control \cite{qualitycontrol}. 

We now, informally, define a good core.

\begin{definition}[Good core (informal; see \Cref{def:goodCore})]\label{def:goodCore-intro}
    For a partition $D_1, D_2, \dots, D_k$ of a (sub)graph $D$, we say the partition satisfies \textsc{Degree-Codegree} if, for each pair $D_i, D_j$, each vertex $u \in D_i$ has degree $\approx |D_j|/2$ into $D_j$, and each pair $u, v \in D_i$ has $\approx |D_j|/4$ shared neighbors in $D_j$.

    We say a (sub)graph $D$ is a \emphdef{good core} if \textsc{Kucera1995} produces such a coloring in expected time $O(|D|^2/k)$, each color class produced by \textsc{Kucera1995} has size $\approx |D|/k$, and the color classes of $D$ satisfy \textsc{Degree-Codegree}, which is sufficient for unique $k$-colorability.
\end{definition}

We say that a graph $G$ has \emphdef{many good cores} if a positive fraction of the subgraphs of size $\widetilde{\Theta}(k^9)$ are good cores. 

Recall that, for a core $D$ to be useful, a large fraction of vertices should have their color inferred by edges to $D$. To formalize this, we define $H$-inferred and strongly $H$-inferred colorings of a subgraph $H$. For a vertex $v$, let $\Gamma(v)$ be its neighborhood.

\begin{definition}\label{def:inferred-coloring}
    Let $H$ be a properly $k$-colored subgraph of a graph $G$, with color classes $\{H_1, H_2, \dots, H_k\}$. We say that a vertex $v \in G$ has an \emphdef{$H$-inferred coloring} if $|\Gamma(v) \cap H_i| \geq 1$ for all but one $i \in [k]$. We say $v$ has a \emphdef{strongly $H$-inferred coloring} if $|\Gamma(v) \cap H_i|/|H_i| \geq 0.01$ for all but one $i \in [k]$.
\end{definition}

If a vertex $v$ has an $H$-inferred coloring, then we can color $v$ by checking its adjacencies with all $w \in H$, which are already colored. If a vertex $v$ has a strongly $H$-inferred coloring and $|H_i| \approx |H|/k$ for all $i \in [k]$, we can color $v$ in $O(k)$ expected time by checking adjacencies with random $w \in H_i$. 
We remark that we will use the guarantees of $H$-inferred and strongly $H$-inferred colorings to build an average-case LCA for $k$-coloring (see \Cref{sec:LCA-implementation}). 

\subsubsection{The density of awesome and okay graphs}
We now define awesome and okay graphs. We define a property called $(\sigma, B)$-linked, which will capture both awesome and okay graphs by choosing different parameters. In words, a $(\sigma, B)$-linked graph has the property that the coloring of \textit{any} set $D$ of size in $[\sigma k, O(k^2 \log k)]$ with approximately balanced color classes can be propagated to all but $B$ vertices in the graph. This is accomplished in two phases by the definition of linked: First, a $0.99$ fraction of vertices in each color class will have their color determined by adjacencies to all but one color class of $D$. All but $B$ of the remaining vertices will have many adjacencies to vertices already colored as all but one remaining color.

\begin{definition}[$(\sigma, B)$-linked (informal; see~\Cref{def:linked})]
   We say that a graph $G$ is $(\sigma, B)$-linked if, for some approximately balanced valid coloring $(A_1, A_2, \ldots, A_k)$, every collection of sets $D_i \subseteq A_i$ of size in $[\sigma, O(k \log k)]$ has the following properties:
   \begin{enumerate}
       \item Let $D = \bigcup_{i \in [k]} D_i$. First, a $0.99$ fraction of vertices in each color class $A_i$, $i \in [k]$ have a strongly $D$-inferred coloring; call the coloring of these vertices $S_1, S_2, \dots, S_k$, and $S := \bigcup_{i \in [k]} S_i$.
       \item Second, all but $B$ vertices in the graph have a strongly $S$-inferred coloring.
   \end{enumerate}
\end{definition}
A graph $G$ is awesome if it has many good cores and is $(\sigma=O(\log k),B=0)$-linked, and okay if it has many good cores and is $(\sigma=O(k\log k),B=O(k\log k))$-linked. 

The structural theorem below is a tail-bound on graphs not being $(\sigma,B)$-linked, which will enable our coloring algorithm to run quickly. Let $\cUnk$ denote the uniform distribution over labeled $k$-colorable graphs with vertex set $[n]$.

\begin{theorem}[Structural Theorem (informal; see \Cref{thm:good_property_probability})]\label{thm:good_property_probability-intro}
    Let $c_1, c_2$ be some specified constants, where $c_1 \in (0, 1)$ and $c_2 \in \mathbb{N}$.
    
    Assume $k\le n^{c_1}$. Then, for every $B\in \N$ and $\sigma\ge c_2 \log k$, with probability at least 
    \[1-\exp(-\Omega(n\sigma/k))-\exp(-\Omega(n(B+1)/k))- \exp(-\Omega(n \log k))\]
    $G\sim \cUnk$ is $(\sigma,B)$-linked (\Cref{def:linked}) and has many good cores.
\end{theorem}

While our main contributions are algorithmic, we briefly discuss the proof of~\Cref{thm:good_property_probability-intro}. To prove \Cref{thm:good_property_probability-intro}, we begin by considering the following family of distributions $\cBnA$ over $k$-colorable graphs: Consider a partition $A:[n]\ra[k]$ such that the preimages of each $i \in [k]$ under $A$ have size approximately $n/k$ (we call this ``balanced''). To sample a graph from the distribution $\cBnA$, for each pair $(u, v) \in A_i \times A_j$ for $i \neq j \in [k]$, add an edge between $u$ and $v$ independently with probability $1/2$. To prove that $G \sim \cBnA$ has many good cores and is $(\sigma, B)$-linked with high probability, we use a standard concentration analysis. 
Ultimately, we are interested not in the distribution $\cBnA$, but in the uniform distribution over $k$-colorable graphs, $\cUnk$. To transfer our results to the uniform distribution, we \textit{simplify a reduction between the distributions} that was given by Dyer-Frieze \cite{dyer1989solution} (and appears in alternative forms in other papers, such as \cite{promel1995random}). 
To transfer from one model to the other, we essentially show that $G\sim \cUnk$ is extremely unlikely to contain an imbalanced coloring (i.e., one where one set is much larger than the others), and the distribution conditioned on \textit{every} balanced coloring is likewise good.

\subsection{Making the Algorithm Deterministic}\label{sec:int:derand}

Prior papers on average-case $k$-coloring~\cite{dyer1989solution,kuvcera1977expected} use separate sections of the input graph as a source of randomness, while we directly argue that the algorithm can fix all random bits by appealing to symmetry properties of the input distribution. 

First, let us look at which steps of the algorithm described so far were randomized. The two randomized steps are sampling random vertices to find a good core (which is certifiably uniquely and efficiently colorable with a balanced coloring), and testing adjacencies to the sets $S_j$ (the already-colored vertices in color class $j$, which are either in the core or colorable via adjacencies to the core) via sampling. Let us begin with the first part of choosing a good core deterministically.

To choose a good core deterministically, since our input is average-case, we can use the inherent randomness of this distribution to argue we can \say{fix} the random choices, while preserving the average runtime. 
Instead of randomly sampling vertices to be in the good core, the algorithm considers a fixed vertex-disjoint collection of possible cores $\cP=(P_1,\ldots,P_{y_0})$ for $y_0 = \tO(k^2)$.
We prove that one of the first $i$ such cores is good with probability $1-\exp(-\Omega(i))$.
This follows from the \textit{vertex symmetry} of the distribution over random $k$-colorable graphs, which is the property that permuting the vertices does not change the probability of sampling a graph under any distribution we consider. 

Our algorithm, therefore, will now enumerate over the $y_0$ vertex-disjoint subsets. If any of these is a good core, we use it; otherwise, we exhaustively search over all sets of size $\nu$ for some $\nu = \text{poly}(k)$ in an arbitrary order to find a core to use. This exhaustive step is inefficient, but it is run with sufficiently low probability, so we can tolerate it.

Once the algorithm has found a good core, it constructs the sets $S_1,\ldots,S_k$ based on adjacencies to the core in a deterministic fashion -- for every $v$ we enumerate over all possible edges to the core. 
The final step we must modify is, for a vertex $v$ that is not yet colored, testing if $|\Gamma(v)\cap S_j|$ is nonzero for all but one $S_j$. 
We again fix a set of vertices $\cL$, and for every vertex $v$, enumerate over $u\in \cL$ and test if $(v,u)\in E$ and $u\in S_j$ for some $S_j$. 
If we find witnesses $u_j$ for every $j\neq i$ in $\cL$, we color $v$ with color $i$. Proving that, with high probability, we find such witnesses in $\cL$ follows the same strategy as before, where we exploit the vertex symmetry of the relevant graph distribution.

In fact, as described, this algorithm will give a runtime of $n\cdot k^c$ for $c\ge 4$. 
Our actual algorithm gradually increases the number of adjacencies we test (i.e., the size of $\cL$) to obtain an expected $\tO(nk)$ runtime. That is, as the size of $\cL$ increases, the success probability increases; by gradually considering larger sets $\cL$ we achieve the right balance between runtime and success probability.

\subsection{Lower Bound}\label{sec:lb-intro-overview}

In \Cref{sec:LB}, we show that any $k$-coloring algorithm that is worst-case correct (i.e., on all inputs and random strings it returns a valid coloring) must make at least $\Omega(nk)$ queries on \textit{every} input graph. Since we require a worst-case correctness guarantee, this immediately implies a lower bound of $\Omega(nk)$ queries on average.

The lower bound relies on a simple counting argument. 
If an algorithm examines only $o(n \cdot k)$ entries of the adjacency structure, then by an averaging argument a constant fraction of the vertices are under-determined --- that is, each is incident to fewer than $k-2$ revealed edges. If the algorithm assigns the coloring $C:[n]\ra[k]$, within one color class there must exist two under-determined vertices whose mutual edge was never queried. 
An adversary can essentially add this single unqueried edge: the resulting graph remains $k$-colorable,\footnote{Our argument is slightly more complicated since this does not hold exactly as written.} but every valid coloring must assign different colors to these two vertices. 
Since the algorithm makes the same decisions on this modified graph, it must output the same coloring $C$, contradicting its worst-case correctness.

\section{Preliminaries}\label{sec:prelims}

\subsection{Prior Results in Average k-Coloring of Graphs}

We recall prior work on coloring that we use, both as subroutines and for the analysis. Let $\cUnk$ be the uniform distribution over $k$-colorable graphs on $n$ vertices. Throughout, graphs are \emph{labeled}: the vertex set is $[n]$, and two isomorphic graphs with different vertex labelings are counted as distinct. $\cUnk$ is uniform over labeled $k$-colorable graphs on $[n]$.

We will utilize that when $k \le n^{1/2}$, the expected number of colorings is bounded, as proven by Dyer and Frieze \cite{dyer1989solution}:

\begin{theorem}[Theorem 3.1~\cite{dyer1989solution}]\label{thm:expColorBounded}
    Let $k\le n^{1/2}$. Let $NC(G)$ be the number of $k$-colorings of $G$ (up to isomorphism). Then $\E_{G\sim \cUnk}[NC(G)]\le 2$.
\end{theorem}
This result is stated for $k$ constant in~\cite{dyer1989solution}, but the proof holds for $k\le n^{1/2}$.

In our analysis, we will work with an alternative model of random $k$-colorable graphs that is easier to analyze, defined in \Cref{def:mod1}. We work with partitions $A:[n]\ra[k]$ (which we think of as colorings). Denote $A_i=A^{-1}(i)$ and $a_i=|A_i|$.
We define partitions (i.e., colorings) to be approximately balanced or weakly balanced as follows:
\begin{definition}\label{def:apx-balanced}
    We say a partition $A:[n]\ra[k]$ is \emphdef{approximately balanced} if $a_i \in (1 \pm \frac{1}{1000k}) n / k$ for every $i$. We say it is \emphdef{weakly balanced} if $a_i \in (1 \pm \frac{1}{100k}) n / k$ for every $i$. 
\end{definition}

For a partition $A:[n]\ra[k]$, define distribution $\cBnA$ as follows.

\begin{definition}[See Model 1 of \cite{dyer1989solution}]\label{def:mod1}
   For every partition $A:[n]\ra[k]$ let $\cBnA$ be the following distribution over graphs. For each pair $(u, v) \in A_i \times A_j$ for $i \neq j \in [k]$, add an edge between $u$ and $v$ independently with probability $1/2$. 
\end{definition}

More generally, for any $m \in \N$ and partition $A:[m]\ra[k]$, let $\cB_{m,k,A}$ be the analogous distribution over graphs on $m$ vertices (so that $\cBnA = \cB_{n,k,A}$). 

Next, we recall a $k$-coloring algorithm of Ku\v{c}era with average time $O(n^2 / k)$.
\begin{theorem}[Theorem 5.7~\cite{kuvcera1995expected}]\label{thm:kuvcera}
    Let $n, k$ be such that $k = o\left(\sqrt{n/\log^2(n)}\right)$. For every weakly balanced partition $A:[n]\ra[k]$, there exists a deterministic algorithm for $k$-coloring $G\sim \cBnA$ with average runtime $O(n^2 / k)$.  Call this algorithm \textsc{Kucera1995}.
\end{theorem}
We use this result to quickly color the small core graph that we sample. If we instead applied a worst-case coloring algorithm (or searched for a core graph that is trivial to quickly color, analogously to the approach of~\cite{hypergraphcoloring}), we would obtain an exponential ($2^{\text{poly}(k)}$) dependence on $k$.

\subsection{Graph Regularity}

We will utilize the notion of $\varepsilon$-regularity, defined below.

\begin{definition}[$\varepsilon$-regular]
    A pair $(A, B)$ is $\varepsilon$-regular if for every $X \subset A$ and $Y \subset B$ satisfying $|X| \geq \varepsilon |A|$ and $|Y| \geq \varepsilon |B|$, we have:
$\left|\frac{e(A, B)}{|A||B|} - \frac{e(X, Y)}{|X||Y|} \right| \leq \varepsilon.$
\end{definition}

\begin{definition}
    For two vertices $u, v$ in a graph, the codegree of $u$ and $v$ is the number of shared neighbors between these two vertices. 
\end{definition}

\cite{DBLP:journals/jal/AlonDLRY94} proved that consistency of vertex degrees and codegrees implies $\varepsilon$-regularity.

\begin{lemma}[Modification/implication of Lemma 3.2 of \cite{DBLP:journals/jal/AlonDLRY94}]\label{regularity-from-codegree}
    Let $C' = (C_1, C_2)$ be a bipartite graph with classes $|C_1|, |C_2| \in (1 \pm 1/(100k)) s/k$. Consider $2 (k/s)^{1/4} < \varepsilon < 1/16$. If at most $(1 - 1/(100k)) \cdot \varepsilon^4 s/(8k) $ vertices of each of $C_1, C_2$ have degree not in $(1 \pm \varepsilon^4)s/(2k)$ and every pair of vertices $(u, v) \in (C_1 \times C_1) \cup (C_2 \times C_2)$ have codegree in $(1 \pm \varepsilon^3)s/(4k)$, then $C'$ is $\varepsilon$-regular.
\end{lemma}

\subsection{Concentration Inequalities}

We will make use of the following concentration inequality, called McDiarmid's Inequality.

\begin{theorem}[McDiarmid's Inequality \cite{mcdiarmid1989method}, lower tail]\label{thm:mcdiarmid}
    Consider a function $f: \mathcal{X}_1 \times \mathcal{X}_2 \times \dots \times \mathcal{X}_n \to \mathbb{R}$. Suppose that the function $f$ satisfies the following \textit{bounded differences} property with constants $(c_i)_{i \in [n]}$: for $i \in [n]$ and all $(x_1, x_2, \dots, x_n) \in \mathcal{X}_1 \times \mathcal{X}_2 \times \dots \times \mathcal{X}_n$, 
    $$\sup_{y_i \in \mathcal{X}_i} \left|f(x_1, x_2, \dots, x_{i - 1}, x_i, x_{i + 1}, \dots,  x_n) - f(x_1, x_2, \dots, x_{i - 1}, y_i, x_{i + 1}, \dots,  x_n) \right| \leq c_i.$$

    Consider random variables $X_i \in \mathcal{X}_i$ for $i  \in [n]$. Then, 
    $$\mathbb{P}\left[f(X_1, X_2, \dots, X_n) \leq \mathbb{E}\left[ f(X_1, X_2, \dots, X_n) \right] - \varepsilon \right] \leq \exp\left( - \frac{2 \varepsilon^2}{\sum_{i = 1}^n c_i^2} \right).$$
\end{theorem}

\section{Properties of Most k-Colorable Graphs}\label{sec:goodProp}

In this section, we present the formal statement of the key technical lemma (see \Cref{thm:good_property_probability}, informally described in \Cref{thm:good_property_probability-intro}) contributing to \Cref{thm:mainRand,thm:mainDet}. The technical lemma says that some strong structural properties hold in all but an exponentially small fraction of $k$-colorable graphs, with respect to the uniform distribution over $k$-colorable graphs.

We additionally prove that subgraphs with certain properties (balanced color classes, and approximate degree and codegree regularity) are uniquely colorable (meaning that the coloring is unique up to permutation of the colors), in \Cref{sec:good-core-colorable}.

\subsection{The properties}

We first formally define these properties, starting with a good core.

\begin{definition}\label{def:goodCore}
    We say a (sub)graph $H$ is a \emphdef{good core} if the following hold. Let $\varepsilon$ be $1/(10k)$.
    \begin{itemize}
        \item $|H| \geq c_0 k^5$, where $c_0 = 20^4$.
        \item The algorithm of~\Cref{thm:kuvcera} (from \cite{kuvcera1995expected}) produces a coloring in time $O(|H|^2/k)$.
        \item Each color class in the coloring produced has size in $(1 \pm \frac{1}{100k}) |H|/k$.
        \item For every pair of independent sets $H_i, H_j \subseteq H$, and every $u \in H_i$, $u$ has degree $(1 \pm \varepsilon^4)|H|/(2k)$ into $H_j$. 
        \item For every pair of independent sets $H_i, H_j \subseteq H$, each pair of vertices in $H_i$ has codegree in $(1 \pm \varepsilon^3) |H|/(4k)$ into $H_j$.
    \end{itemize}
\end{definition}

We prove in \Cref{lem:core-unique-color} that such an $H$ is \textit{uniquely colorable}.

We say a graph has many good cores if a positive fraction of subgraphs (of a fixed size) are good cores: 

\begin{definition}\label{def:manyGoodCores}
    Fix $\nu \defeq \lceil 2c_\nu k^9 \log k\rceil$, where $c_\nu>0$ is a sufficiently large
    constant (chosen so that \Cref{lem:num-good-subgraphs,lem:manyGoodCores} hold). We say a graph
    $G$ has \emphdef{many good cores} if at least $2/3$ of its subgraphs of size $\nu$ are good
    cores (\Cref{def:goodCore}).
\end{definition}

Finally, we say a graph is linked if, essentially, every core contains edges to most vertices, and all but a tiny number of exceptional vertices have many edges to the former subset.

\begin{restatable}{definition}{linked}\label{def:linked}
    We say $G$ is $(\sigma,B)$\emphdef{-linked} if there is an approximately balanced valid coloring $(A_1, A_2, \ldots,A_k)$ such that the following holds. For every collection of sets $D_i\subseteq A_i$ with $|D_i|\in [\sigma,O(k \log k)]$, the following two properties hold:
    \begin{enumerate}
        \item\label{itm:core-inferred} For every $i$, at least 99\% of the vertices in $A_i$ are strongly $(D_1,\ldots,D_k)$-inferred (\Cref{def:inferred-coloring}).

        \item\label{itm:s-inferred} For every collection of sets $S_i\subseteq A_i$ where $|S_i|/|A_i|\ge 0.9$, all but $B$ vertices are strongly $(S_1,\ldots,S_k)$-inferred.
    \end{enumerate}
\end{restatable}

\subsection{A good core is uniquely colorable}\label{sec:good-core-colorable}

We prove that good cores are uniquely colorable. The proof of this lemma is where regularity plays a crucial role. Regularity implies that the simple, efficiently checkable conditions of a good core certify the unique colorability of a good core.

\begin{lemma}\label{lem:core-unique-color}
    If a (sub)graph is a good core, it is uniquely colorable.
\end{lemma}
\begin{proof}
    
    Let $H$ be a good core on $s$ vertices, and let $C = (C_1, C_2, \dots, C_k)$ be the output of the algorithm of \Cref{thm:kuvcera} (from \cite{kuvcera1995expected}) on $H$ that satisfies the properties of an acceptable $k$-coloring of a good core. Suppose, towards contradiction, that there is another proper $k$-coloring $D = (D_1, D_2, \dots, D_k)$ of $H$. This coloring $D$ does not necessarily satisfy the conditions (balancedness, degrees, codegrees) of \Cref{def:goodCore}, so we cannot make any further assumptions about it.

    First, we need a lower bound on the size of some $D_i$ and its intersection with some color classes $C_j$ and $C_{j'}$. We prove that there exists a color class $D_i$ such that $|D_i| > (1 - \frac{k-1}{100k}) s /k$ and there are $j, j' \in [k]$ with $|C_j \cap D_i| > (1 - \frac{k-1}{100k}) s /k^2$ and $|C_{j'} \cap D_i| \geq 1$. We consider two cases. 
    \begin{enumerate}
        \item Case A: Suppose there is a color class $D_i$ of size $> (1 + \frac{1}{100k}) s /k$. Since each color class of $C$ has size in $(1 \pm \frac{1}{100k}) s /k$, $D_i$ cannot be contained in only one color class of $C$. By the pigeonhole principle over the sizes of the intersections of color classes of $C$ with $D_i$, there must exist color classes $C_j, C_{j'}$ such that $|C_j \cap D_i| \geq |D_i|/k > (1 + \frac{1}{100k}) s /k^2$ and $|C_{j'} \cap D_i| \geq 1$.
        \item Case B: Suppose that all color classes of $D$ are of size $< (1 + \frac{1}{100k}) s /k$. Then, each color class of $D$ must be of size $> (1 - \frac{k-1}{100k}) s /k$. Since $D \neq C$ (as partitions, i.e., up to relabeling of colors), there exists a $D_i$ such that $D_i$ crosses at least two color classes of $C$. By the pigeonhole principle, there must exist color classes $C_j, C_{j'}$ such that $|C_j \cap D_i| \geq |D_i|/k > (1 - \frac{k-1}{100k}) s /k^2$ and $|C_{j'} \cap D_i| \geq 1$.
    \end{enumerate}
    Thus, we have proven the required statement.
    
    Next, without loss of generality, suppose $D$ satisfies $|D_1| \geq 0.99 s/k$, $|C_1 \cap D_1| \geq 0.99 s/k^2$, and $|C_2 \cap D_1| \geq 1$, and let $C_1$ and $C_2$ be the color classes of $C$ with the largest and second largest intersection size with $D_1$. We consider two cases. 
    \\\\
    \noindent\textit{Case 1.} 
    First, we consider the case where $C_1$ and $D_1$ agree a lot, but not completely. Suppose $|C_1 \cap D_1| \geq  0.8 s/k$. This contradicts the degree condition required for a good core, as vertices $v \in C_2 \cap D_1$ \textit{cannot} be connected to any vertices in $C_1 \cap D_1$, and thus cannot have a degree of more than $(1+\frac{1}{100k})s/k - 0.8 s/k = s/k \cdot (1 + \frac{1}{100k} - 0.8) = (0.2 + \frac{1}{100k})s/k$ into $C_1$. This is much smaller than $(1 - \varepsilon^4) s/(2k)$, which is a contradiction.

    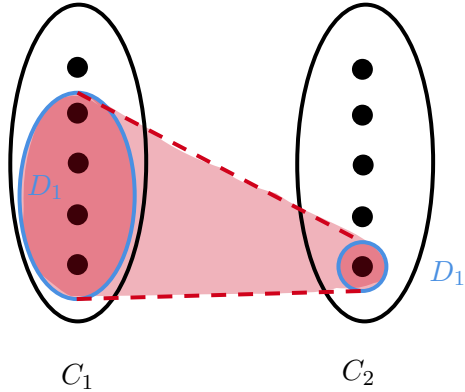
\begin{figure}[H]
        \centering
        \tikzset{every picture/.style={line width=0.75pt}} 

\begin{tikzpicture}[x=0.75pt,y=0.75pt,yscale=-1,xscale=1]

\draw [draw opacity=0][fill={rgb, 255:red, 208; green, 2; blue, 27 }  ,fill opacity=0.3 ][line width=3] [line join = round][line cap = round]   (257.75,59.5) .. controls (267.5,63) and (273.5,69) .. (288.5,76) .. controls (290.25,80.5) and (310.29,91) .. (323.29,95) .. controls (331.29,102) and (340.5,105) .. (351.5,108) .. controls (362.5,116) and (376.32,120.91) .. (388.5,127) .. controls (397.25,131.5) and (405.5,137) .. (406.75,145.5) .. controls (405.94,148.76) and (406.5,152) .. (400.5,153) .. controls (393.5,156) and (387.5,156) .. (381.75,155.5) .. controls (372.15,158.7) and (361.29,157.33) .. (350.75,158.5) .. controls (322.5,158) and (282.75,159.84) .. (248.75,159.5) .. controls (246.05,159.47) and (236.55,151.04) .. (234.5,149.5) .. controls (228.5,144.5) and (226.15,136.73) .. (224.5,128.5) .. controls (223.5,131.5) and (224.86,117.64) .. (224.75,114.5) .. controls (224.02,93.28) and (228.47,79.06) .. (235.5,65) .. controls (242.5,60) and (251.5,55) .. (262.5,64) ;
\draw  [line width=1.5]  (252.45,13.3) .. controls (271.29,13.37) and (286.42,49.02) .. (286.24,92.92) .. controls (286.07,136.82) and (270.65,172.35) .. (251.82,172.27) .. controls (232.98,172.2) and (217.85,136.55) .. (218.03,92.65) .. controls (218.2,48.75) and (233.62,13.22) .. (252.45,13.3) -- cycle ;
\draw  [line width=1.5]  (396.45,13.3) .. controls (415.29,13.37) and (430.42,49.02) .. (430.24,92.92) .. controls (430.07,136.82) and (414.65,172.35) .. (395.82,172.27) .. controls (376.98,172.2) and (361.85,136.55) .. (362.03,92.65) .. controls (362.2,48.75) and (377.62,13.22) .. (396.45,13.3) -- cycle ;
\draw  [fill={rgb, 255:red, 0; green, 0; blue, 0 }  ,fill opacity=1 ] (247,44.75) .. controls (247,42.13) and (249.13,40) .. (251.75,40) .. controls (254.37,40) and (256.5,42.13) .. (256.5,44.75) .. controls (256.5,47.37) and (254.37,49.5) .. (251.75,49.5) .. controls (249.13,49.5) and (247,47.37) .. (247,44.75) -- cycle ;
\draw  [fill={rgb, 255:red, 0; green, 0; blue, 0 }  ,fill opacity=1 ] (247,118.75) .. controls (247,116.13) and (249.13,114) .. (251.75,114) .. controls (254.37,114) and (256.5,116.13) .. (256.5,118.75) .. controls (256.5,121.37) and (254.37,123.5) .. (251.75,123.5) .. controls (249.13,123.5) and (247,121.37) .. (247,118.75) -- cycle ;
\draw  [fill={rgb, 255:red, 0; green, 0; blue, 0 }  ,fill opacity=1 ] (247.38,92.79) .. controls (247.38,90.16) and (249.51,88.04) .. (252.13,88.04) .. controls (254.76,88.04) and (256.88,90.16) .. (256.88,92.79) .. controls (256.88,95.41) and (254.76,97.54) .. (252.13,97.54) .. controls (249.51,97.54) and (247.38,95.41) .. (247.38,92.79) -- cycle ;
\draw  [fill={rgb, 255:red, 0; green, 0; blue, 0 }  ,fill opacity=1 ] (247,67.75) .. controls (247,65.13) and (249.13,63) .. (251.75,63) .. controls (254.37,63) and (256.5,65.13) .. (256.5,67.75) .. controls (256.5,70.37) and (254.37,72.5) .. (251.75,72.5) .. controls (249.13,72.5) and (247,70.37) .. (247,67.75) -- cycle ;
\draw  [fill={rgb, 255:red, 0; green, 0; blue, 0 }  ,fill opacity=1 ] (247,143.75) .. controls (247,141.13) and (249.13,139) .. (251.75,139) .. controls (254.37,139) and (256.5,141.13) .. (256.5,143.75) .. controls (256.5,146.37) and (254.37,148.5) .. (251.75,148.5) .. controls (249.13,148.5) and (247,146.37) .. (247,143.75) -- cycle ;
\draw  [fill={rgb, 255:red, 0; green, 0; blue, 0 }  ,fill opacity=1 ] (390,45.75) .. controls (390,43.13) and (392.13,41) .. (394.75,41) .. controls (397.37,41) and (399.5,43.13) .. (399.5,45.75) .. controls (399.5,48.37) and (397.37,50.5) .. (394.75,50.5) .. controls (392.13,50.5) and (390,48.37) .. (390,45.75) -- cycle ;
\draw  [fill={rgb, 255:red, 0; green, 0; blue, 0 }  ,fill opacity=1 ] (390,119.75) .. controls (390,117.13) and (392.13,115) .. (394.75,115) .. controls (397.37,115) and (399.5,117.13) .. (399.5,119.75) .. controls (399.5,122.37) and (397.37,124.5) .. (394.75,124.5) .. controls (392.13,124.5) and (390,122.37) .. (390,119.75) -- cycle ;
\draw  [fill={rgb, 255:red, 0; green, 0; blue, 0 }  ,fill opacity=1 ] (390.38,93.79) .. controls (390.38,91.16) and (392.51,89.04) .. (395.13,89.04) .. controls (397.76,89.04) and (399.88,91.16) .. (399.88,93.79) .. controls (399.88,96.41) and (397.76,98.54) .. (395.13,98.54) .. controls (392.51,98.54) and (390.38,96.41) .. (390.38,93.79) -- cycle ;
\draw  [fill={rgb, 255:red, 0; green, 0; blue, 0 }  ,fill opacity=1 ] (390,68.75) .. controls (390,66.13) and (392.13,64) .. (394.75,64) .. controls (397.37,64) and (399.5,66.13) .. (399.5,68.75) .. controls (399.5,71.37) and (397.37,73.5) .. (394.75,73.5) .. controls (392.13,73.5) and (390,71.37) .. (390,68.75) -- cycle ;
\draw  [fill={rgb, 255:red, 0; green, 0; blue, 0 }  ,fill opacity=1 ] (390,144.75) .. controls (390,142.13) and (392.13,140) .. (394.75,140) .. controls (397.37,140) and (399.5,142.13) .. (399.5,144.75) .. controls (399.5,147.37) and (397.37,149.5) .. (394.75,149.5) .. controls (392.13,149.5) and (390,147.37) .. (390,144.75) -- cycle ;
\draw  [color={rgb, 255:red, 74; green, 144; blue, 226 }  ,draw opacity=1 ][fill={rgb, 255:red, 208; green, 2; blue, 27 }  ,fill opacity=0.3 ][line width=1.5]  (251.77,57.51) .. controls (267.73,57.57) and (280.57,80.8) .. (280.46,109.4) .. controls (280.34,137.99) and (267.31,161.12) .. (251.36,161.05) .. controls (235.4,160.99) and (222.56,137.76) .. (222.67,109.17) .. controls (222.79,80.57) and (235.82,57.45) .. (251.77,57.51) -- cycle ;
\draw  [color={rgb, 255:red, 74; green, 144; blue, 226 }  ,draw opacity=1 ][fill={rgb, 255:red, 208; green, 2; blue, 27 }  ,fill opacity=0.3 ][line width=1.5]  (394.8,132.5) .. controls (401.44,132.52) and (406.81,138.03) .. (406.78,144.8) .. controls (406.75,151.57) and (401.34,157.03) .. (394.7,157) .. controls (388.06,156.98) and (382.69,151.47) .. (382.72,144.7) .. controls (382.75,137.93) and (388.16,132.47) .. (394.8,132.5) -- cycle ;
\draw [color={rgb, 255:red, 208; green, 2; blue, 27 }  ,draw opacity=1 ][line width=1.5]  [dash pattern={on 5.63pt off 4.5pt}]  (251.77,57.51) -- (394.8,132.5) ;
\draw [color={rgb, 255:red, 208; green, 2; blue, 27 }  ,draw opacity=1 ][line width=1.5]  [dash pattern={on 5.63pt off 4.5pt}]  (251.36,161.05) -- (394.7,157) ;

\draw (242,192.4) node [anchor=north west][inner sep=0.75pt]    {$C_{1}$};
\draw (383,190.4) node [anchor=north west][inner sep=0.75pt]    {$C_{2}$};
\draw (225,97.4) node [anchor=north west][inner sep=0.75pt]  [color={rgb, 255:red, 74; green, 144; blue, 226 }  ,opacity=1 ]  {$D_{1}$};
\draw (427,140.4) node [anchor=north west][inner sep=0.75pt]  [color={rgb, 255:red, 74; green, 144; blue, 226 }  ,opacity=1 ]  {$D_{1}$};

\end{tikzpicture}
        \caption{Case 1: Towards contradiction, suppose there are two $k$-colorings $C = (C_1, C_2, \dots, C_k)$ and $D = (D_1, D_2, \dots, D_k)$ of a good core. Suppose $|C_1 \cap D_1| \geq 0.8 s/k$ and $|C_2 \cap D_1| \geq 1$. Every vertex $v \in C_2 \cap D_1$ cannot have any edges to $C_1 \cap D_1$. Since $C_1 \setminus D_1$ is so small, $v \in C_2 \cap D_1$ contradicts the degree requirement imposed on vertices in a good core. In the figure, the red region represents where edges are forbidden.}
        \label{fig:good-core-case-1}
    \end{figure}
    
    \bigskip

    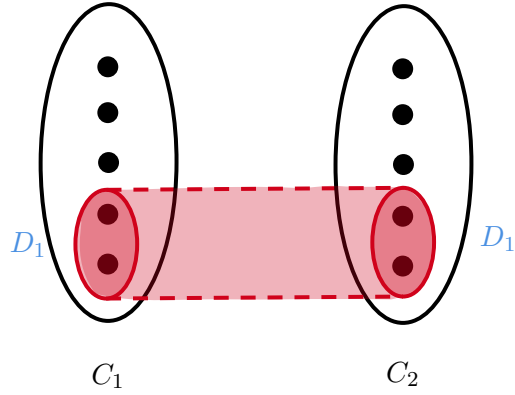
\begin{figure}[h]
        \centering
        \tikzset{every picture/.style={line width=0.75pt}} 

\begin{tikzpicture}[x=0.75pt,y=0.75pt,yscale=-1,xscale=1]

\draw [color={rgb, 255:red, 208; green, 2; blue, 27 }  ,draw opacity=1 ][line width=1.5]  [dash pattern={on 5.63pt off 4.5pt}]  (254.1,104.49) -- (403.1,103.49) ;
\draw  [line width=1.5]  (255.45,11.3) .. controls (274.29,11.37) and (289.42,47.02) .. (289.24,90.92) .. controls (289.07,134.82) and (273.65,170.35) .. (254.82,170.27) .. controls (235.98,170.2) and (220.85,134.55) .. (221.03,90.65) .. controls (221.2,46.75) and (236.62,11.22) .. (255.45,11.3) -- cycle ;
\draw  [line width=1.5]  (403.45,12.3) .. controls (422.29,12.37) and (437.42,48.02) .. (437.24,91.92) .. controls (437.07,135.82) and (421.65,171.35) .. (402.82,171.27) .. controls (383.98,171.2) and (368.85,135.55) .. (369.03,91.65) .. controls (369.2,47.75) and (384.62,12.22) .. (403.45,12.3) -- cycle ;
\draw  [fill={rgb, 255:red, 0; green, 0; blue, 0 }  ,fill opacity=1 ] (250,42.75) .. controls (250,40.13) and (252.13,38) .. (254.75,38) .. controls (257.37,38) and (259.5,40.13) .. (259.5,42.75) .. controls (259.5,45.37) and (257.37,47.5) .. (254.75,47.5) .. controls (252.13,47.5) and (250,45.37) .. (250,42.75) -- cycle ;
\draw  [fill={rgb, 255:red, 0; green, 0; blue, 0 }  ,fill opacity=1 ] (250,116.75) .. controls (250,114.13) and (252.13,112) .. (254.75,112) .. controls (257.37,112) and (259.5,114.13) .. (259.5,116.75) .. controls (259.5,119.37) and (257.37,121.5) .. (254.75,121.5) .. controls (252.13,121.5) and (250,119.37) .. (250,116.75) -- cycle ;
\draw  [fill={rgb, 255:red, 0; green, 0; blue, 0 }  ,fill opacity=1 ] (250.38,90.79) .. controls (250.38,88.16) and (252.51,86.04) .. (255.13,86.04) .. controls (257.76,86.04) and (259.88,88.16) .. (259.88,90.79) .. controls (259.88,93.41) and (257.76,95.54) .. (255.13,95.54) .. controls (252.51,95.54) and (250.38,93.41) .. (250.38,90.79) -- cycle ;
\draw  [fill={rgb, 255:red, 0; green, 0; blue, 0 }  ,fill opacity=1 ] (250,65.75) .. controls (250,63.13) and (252.13,61) .. (254.75,61) .. controls (257.37,61) and (259.5,63.13) .. (259.5,65.75) .. controls (259.5,68.37) and (257.37,70.5) .. (254.75,70.5) .. controls (252.13,70.5) and (250,68.37) .. (250,65.75) -- cycle ;
\draw  [fill={rgb, 255:red, 0; green, 0; blue, 0 }  ,fill opacity=1 ] (250,141.75) .. controls (250,139.13) and (252.13,137) .. (254.75,137) .. controls (257.37,137) and (259.5,139.13) .. (259.5,141.75) .. controls (259.5,144.37) and (257.37,146.5) .. (254.75,146.5) .. controls (252.13,146.5) and (250,144.37) .. (250,141.75) -- cycle ;
\draw  [fill={rgb, 255:red, 0; green, 0; blue, 0 }  ,fill opacity=1 ] (398,43.75) .. controls (398,41.13) and (400.13,39) .. (402.75,39) .. controls (405.37,39) and (407.5,41.13) .. (407.5,43.75) .. controls (407.5,46.37) and (405.37,48.5) .. (402.75,48.5) .. controls (400.13,48.5) and (398,46.37) .. (398,43.75) -- cycle ;
\draw  [fill={rgb, 255:red, 0; green, 0; blue, 0 }  ,fill opacity=1 ] (398,117.75) .. controls (398,115.13) and (400.13,113) .. (402.75,113) .. controls (405.37,113) and (407.5,115.13) .. (407.5,117.75) .. controls (407.5,120.37) and (405.37,122.5) .. (402.75,122.5) .. controls (400.13,122.5) and (398,120.37) .. (398,117.75) -- cycle ;
\draw  [fill={rgb, 255:red, 0; green, 0; blue, 0 }  ,fill opacity=1 ] (398.38,91.79) .. controls (398.38,89.16) and (400.51,87.04) .. (403.13,87.04) .. controls (405.76,87.04) and (407.88,89.16) .. (407.88,91.79) .. controls (407.88,94.41) and (405.76,96.54) .. (403.13,96.54) .. controls (400.51,96.54) and (398.38,94.41) .. (398.38,91.79) -- cycle ;
\draw  [fill={rgb, 255:red, 0; green, 0; blue, 0 }  ,fill opacity=1 ] (398,66.75) .. controls (398,64.13) and (400.13,62) .. (402.75,62) .. controls (405.37,62) and (407.5,64.13) .. (407.5,66.75) .. controls (407.5,69.37) and (405.37,71.5) .. (402.75,71.5) .. controls (400.13,71.5) and (398,69.37) .. (398,66.75) -- cycle ;
\draw  [fill={rgb, 255:red, 0; green, 0; blue, 0 }  ,fill opacity=1 ] (398,142.75) .. controls (398,140.13) and (400.13,138) .. (402.75,138) .. controls (405.37,138) and (407.5,140.13) .. (407.5,142.75) .. controls (407.5,145.37) and (405.37,147.5) .. (402.75,147.5) .. controls (400.13,147.5) and (398,145.37) .. (398,142.75) -- cycle ;
\draw  [color={rgb, 255:red, 208; green, 2; blue, 27 }  ,draw opacity=1 ][fill={rgb, 255:red, 208; green, 2; blue, 27 }  ,fill opacity=0.3 ][line width=1.5]  (254.1,104.49) .. controls (262.66,104.52) and (269.55,116.76) .. (269.49,131.83) .. controls (269.43,146.9) and (262.44,159.09) .. (253.89,159.05) .. controls (245.33,159.02) and (238.44,146.78) .. (238.5,131.71) .. controls (238.56,116.64) and (245.55,104.45) .. (254.1,104.49) -- cycle ;
\draw  [color={rgb, 255:red, 208; green, 2; blue, 27 }  ,draw opacity=1 ][fill={rgb, 255:red, 208; green, 2; blue, 27 }  ,fill opacity=0.3 ][line width=1.5]  (403.1,103.49) .. controls (411.66,103.52) and (418.55,115.76) .. (418.49,130.83) .. controls (418.43,145.9) and (411.44,158.09) .. (402.89,158.05) .. controls (394.33,158.02) and (387.44,145.78) .. (387.5,130.71) .. controls (387.56,115.64) and (394.55,103.45) .. (403.1,103.49) -- cycle ;
\draw [color={rgb, 255:red, 208; green, 2; blue, 27 }  ,draw opacity=1 ][line width=1.5]  [dash pattern={on 5.63pt off 4.5pt}]  (253.89,159.05) -- (402.89,158.05) ;
\draw [draw opacity=0][fill={rgb, 255:red, 208; green, 2; blue, 27 }  ,fill opacity=0.3 ][line width=3] [line join = round][line cap = round]   (252.5,104.5) .. controls (284.5,104.5) and (316.54,101.9) .. (348.5,103.5) .. controls (355.5,103.85) and (362.5,102.24) .. (369.5,102.5) .. controls (382.5,103.5) and (391.15,102.4) .. (403.1,103.49) .. controls (406.83,103.83) and (413.76,112.01) .. (415.5,115.5) .. controls (419.5,127.5) and (417.37,136.16) .. (414.5,150.5) .. controls (413.94,153.32) and (399.83,159.56) .. (396.5,159.5) .. controls (377.5,159.17) and (358.5,159.17) .. (339.5,159.5) .. controls (331.05,159.65) and (308.25,159.28) .. (299.5,159.5) .. controls (278.44,160.04) and (263.5,160.5) .. (248.5,156.5) .. controls (237.5,141.5) and (241.2,126.42) .. (241.5,120.5) .. controls (241.55,119.55) and (239.5,122.5) .. (241.5,114.5) .. controls (245.5,109.5) and (250.28,105.5) .. (253.5,105.5) ;

\draw (245,190.4) node [anchor=north west][inner sep=0.75pt]    {$C_{1}$};
\draw (393,189.4) node [anchor=north west][inner sep=0.75pt]    {$C_{2}$};
\draw (204,124.4) node [anchor=north west][inner sep=0.75pt]  [color={rgb, 255:red, 74; green, 144; blue, 226 }  ,opacity=1 ]  {$D_{1}$};
\draw (440,121.4) node [anchor=north west][inner sep=0.75pt]  [color={rgb, 255:red, 74; green, 144; blue, 226 }  ,opacity=1 ]  {$D_{1}$};

\end{tikzpicture}
        \caption{Case 2: Towards contradiction, suppose there are two $k$-colorings $C = (C_1, C_2, \dots, C_k)$ and $D = (D_1, D_2, \dots, D_k)$ of a good core. Suppose $|C_1 \cap D_1| \in [0.99 s/k^2, 0.8 s/k)$ and $|C_2 \cap D_1| \geq 0.1 s/k^2$. There can be no edges between $C_1 \cap D_1$ and $C_2 \cap D_1$, since these are subsets of the same color class $D_1$ in the coloring $D$. However, by \Cref{regularity-from-codegree} (from \cite{DBLP:journals/jal/AlonDLRY94}), each pair of large subsets in different independent sets (here, $C_1$ and $C_2$) in a good core must have many edges between them. This is a contradiction. In the figure, the red region represents where edges are forbidden.}
        \label{fig:good-core-case-2}
    \end{figure}

    \noindent\textit{Case 2.}  Next, we consider the case where $|C_1 \cap D_1|$ is not only lower-bounded but also upper-bounded. In this case, we will identify that there is another color class $C_2$ such that $|C_2 \cap D_1|$ is large. Recall that, by \Cref{def:goodCore} of a good core, we have verified that each vertex's degree is in some range and each pair of vertices has a codegree in some range. By \Cref{regularity-from-codegree} (from \cite{DBLP:journals/jal/AlonDLRY94}), this implies regularity of $C$ (the good coloring of the good core), which implies that the bipartite graph between any two large enough subsets of different color classes of $C$ must be sufficiently dense.

    We now state all this more formally. Suppose that $|C_1 \cap D_1| \in [0.99 s/k^2, 0.8 s/k)$. By the pigeonhole principle, there is a color class (WLOG $C_2$) such that $|C_2 \cap D_1| \geq (|D_1| - |C_1 \cap D_1|)/k \geq (0.99 s/k - 0.8 s/k)/k > 0.1 s/k^2$. There are no edges between $C_1 \cap D_1$ and $C_2 \cap D_1$. These are two sets in different partition classes, each of size at least $0.1 s/k^2$, that must have \textit{zero} edges between them. Therefore we have both $|C_1 \cap D_1|, |C_2 \cap D_1| \geq 0.1 s/k^2$ and $e(C_1 \cap D_1, C_2 \cap D_1) = 0$. Note that each of these sets has size at least $0.1 s/k^2$, which is at least an $\varepsilon$-fraction of its color class (of size $\approx s/k$), since $\frac{0.1 s/k^2}{s/k} = \frac{1}{10k} = \varepsilon$. Hence they are large enough for $\varepsilon$-regularity to apply.
    
    However, this violates the $\varepsilon$-regularity property, for $\varepsilon = 1/(10k)$: The approximate degree and codegree properties of the good core from \Cref{def:goodCore} together imply $\varepsilon$-regularity (by \Cref{regularity-from-codegree} from \cite{DBLP:journals/jal/AlonDLRY94}), which says that any two reasonably sized subsets behave as if they had the expected edge density $\bar{d}$. To apply \Cref{regularity-from-codegree} with $\varepsilon = 1/(10k)$ we need $2(k/s)^{1/4} < \varepsilon < 1/16$, i.e., $s > (20 k^{5/4})^4 = 20^4 k^5$, which holds since $s = |H| \geq c_0 k^5$ by \Cref{def:goodCore}. 
    
    More formally, by the definition of $\varepsilon$-regular and since $|C_1 \cap D_1|, |C_2 \cap D_1| \geq 0.1 s/k^2$, for $$\bar{d} \geq \frac{(1 - \varepsilon^4)}{2(1 + \frac{1}{100k})}\frac{(s/k)^2}{(s/k)^2} > 0.45,$$ the edge density between $C_1 \cap D_1$ and $C_2 \cap D_1$ must be at least $\bar{d}-\varepsilon>0$.  
    Thus,  we must have 
    $$0 = \frac{e(C_1 \cap D_1, C_2 \cap D_1)}{|C_1 \cap D_1| |C_2 \cap D_1|} \geq \bar{d} - \varepsilon > 0.$$
    That is, there must be zero edges between $C_1 \cap D_1$ and $C_2 \cap D_1$ since they are each subsets of the same color class $D_1$ in $D$. But because $C_1 \cap D_1 \subseteq C_1$ and $C_2 \cap D_1 \subseteq C_2$ are subsets of different color classes of $C$, it must be the case that the edge density between these subsets is at least $\bar{d} - \varepsilon > 0$. Contradiction. 
    
    \medskip\noindent Thus it must be the case that the coloring $D$ equals the coloring $C$. Therefore, when a subgraph is a good core, it is uniquely colorable.
\end{proof}

\subsection{Structural properties for propagating the coloring}\label{sec:structural-propagating}

In this section, we prove \Cref{thm:good_property_probability}, which states that structural properties that allow us to successfully propagate the unique coloring of a good core to most of the remaining graph hold with very high probability.

We bound the probability of a graph having many good cores and being linked as follows.
\begin{theorem}\label{thm:good_property_probability}
    Assume $k\le n^{1/36-\delta}$. There is a constant $c_\sigma>0$ such that
    for every $B\in \N$ and $\sigma\ge c_\sigma\log k$, with probability at least 
    \[1-\exp(-\Omega(n\sigma/k))-\exp(-\Omega(n(B+1)/k))-2k^{-2n}\]
    $G\sim \cUnk$ is $(\sigma,B)$-linked (\Cref{def:linked}) and has many good cores (\Cref{def:manyGoodCores}).
\end{theorem}
To analyze both algorithms, we define two levels of how well-linked a graph is:
\begin{definition}\label{def:awesomeokay}
    We say $G$ is \emphdef{awesome} if it has many good cores and is $(c_\sigma\log k,0)$-linked. We say $G$ is \emphdef{okay} if it has many good cores and is $(ck\log k,ck\log k)$-linked, where $c>0$ is chosen such that the bound in~\Cref{thm:good_property_probability} is at least $1-3k^{-2n}$.
\end{definition}

For readability, we name the parameters of these two definitions:
$$ \sigma_A \defeq c_\sigma\log k, \qquad
    \sigma_O \defeq ck\log k, \qquad
    B_O \defeq ck\log k, $$
so that awesome corresponds to $(\sigma_A,0)$-linked with many good cores, and okay corresponds to
$(\sigma_O,B_O)$-linked with many good cores.

To prove~\Cref{thm:good_property_probability}, we first work with the alternative model $\cBnA$ of \Cref{def:mod1}.
We then transfer from this distribution to the uniform distribution over $k$-colorable graphs with a proof inspired by (but simpler than) that of Dyer-Frieze~\cite{dyer1989solution}.

First, for an approximately balanced $A$ (as in \Cref{def:apx-balanced}), there are many good cores in $G\sim \cBnA$ with extremely high probability. Recall from \Cref{def:goodCore} that a subgraph $S$ is a \textit{good core} if it has color classes that are weakly balanced, the algorithm of~\Cref{thm:kuvcera} (from \cite{kuvcera1995expected}) produces a proper $k$-coloring in expected time $O(|S|^2/k)$, and a sufficient condition for unique colorability (approximate regularity of degrees and codegrees into color classes) holds.

\begin{lemma}\label{lem:num-good-subgraphs}
    Let $A$ be approximately balanced. Over $G \sim \cBnA$, the expected number of size-$\nu$ subgraphs that are good cores is at least $\frac{9}{10} \binom{n}{\nu}$, where 
    $\nu = \Omega( k^9 \log(k))$.
\end{lemma}
\begin{proof}
    Let $H$ be the vertex set of a uniformly random subgraph of size $\nu$, and let $H_i = H \cap A_i$ and $h_i = |H_i|$.

    We first prove that $h_i \in (1 \pm \frac{1}{100k}) \nu/k$ with probability at least $1 - 0.025$. The color-class counts $(h_1,\ldots,h_k)$ are distributed as the multivariate hypergeometric distribution with $\nu$ trials and counts $(a_1,\ldots,a_k)$. Recall that $a_i\in (1 \pm \frac{1}{1000k}) n / k$. By standard concentration results (see e.g.,~\cite{CHVATAL}), with probability at least $1 - 0.025$ we have that $h_i\in (1 \pm \frac{1}{100k}) \nu/k$, when $\nu \geq c k^3 \log k$ for a constant $c$. We say that $H$ is \emphdef{weakly balanced} if this holds (and note that it is equivalent to condition 3 of $H$ being a good core).

    Second, conditioning on the event that $H$ induces some weakly balanced partition $(H_1,\ldots,H_k)$, the conditional distribution of the induced subgraph is $\cB_{\nu,k,H}$ (\Cref{def:mod1}). Thus, the algorithm of~\Cref{thm:kuvcera} produces a coloring in expected time $O(|H|^2/k)$, and by Markov's inequality produces a coloring in $O(|H|^2/k)$ time with probability at least $1 - 0.025$.
 
    Third, we prove that each vertex has degree in $(1 \pm \varepsilon^4) \nu/(2k)$ into each other color class, for $\varepsilon = O(1/k)$, with probability at least $1 - 0.025$. This follows via Chernoff bounds for $\nu \geq c k^2 \log(k)/\varepsilon^8$ for a constant $c$ and $\varepsilon = O(1/k)$, together with a union bound. This is because the expected degree into another color class is approximately $\nu/(2k) \geq k \log (k)$, we are considering an error bound of $\varepsilon^4$, and there are $\nu k$ pairs of vertices and other color classes to consider.

    Fourth, we prove that each pair of vertices in each $H_i$ has codegree in $(1 \pm \varepsilon^3)\nu/(4k)$ into each other color class, for $\varepsilon = O(1/k)$, with probability at least $1 - 0.025$. The proof is the same as for degrees, since the expected codegree is approximately $\nu/(4k) \geq k^2 \log k$, and there are at most $\nu \cdot \nu/k \cdot k$ vertex-color class pairs to consider.

    For a subgraph $H$, let $E_1$ be the event that the subgraph is balanced, $E_2$ be the event that the algorithm of~\Cref{thm:kuvcera} produces a proper $k$-coloring in time $O(|H|^2/k)$, $E_3$ be the event that degrees are almost regular, and $E_4$ be the event that codegrees are almost regular. Each of the four bounds above is on the \emph{marginal} probability of the corresponding event, and each was shown to fail with probability at most $0.025$. By a union bound,
    $$\Pr_H\left(\neg(E_1 \land E_2 \land E_3 \land E_4)\right) \leq \sum_{i=1}^4 \Pr_H(\neg E_i) \leq 4 \cdot 0.025 = 0.1,$$
    so $\Pr_H(H \text{ is a good core}) \geq \Pr_H(E_1 \land E_2 \land E_3 \land E_4) \geq 9/10$. By linearity of expectation, the expected number of subgraphs of size $\nu$ that are good cores is at least $\frac{9}{10} \binom{n}{\nu}$.
\end{proof}

\begin{lemma}\label{lem:manyGoodCores}
    Assume $k\le n^{1/36-\delta}$. 
    For every approximately balanced $A$, $G \sim \cBnA$ has many good cores with probability at least $1 - k^{-2n}$. 
\end{lemma}
\begin{proof}
     We apply McDiarmid's inequality (\Cref{thm:mcdiarmid}) to provide a high-probability lower bound on the number of good cores of size $\nu$ in $G \sim \cBnA$. For each unordered pair $\{u, v\}$ of vertices, define the edge indicator $X_{\{u, v\}} = \mathbf{1}(\{u, v\} \in E(G))$, and let $P = \binom{V(G)}{2}$ be the set of all such pairs. Define $f: \{0,1\}^{P} \to \mathbb{N}$ to be the function that maps a tuple of edge indicators $(x_{\{u,v\}})_{\{u,v\} \in P}$ to the number of good cores of size $\nu$ in the corresponding graph. First, from \Cref{lem:num-good-subgraphs}, we know that the expected number of good cores of size $\nu$ in $G \sim \cBnA$ is at least $\frac{9}{10} \binom{n}{\nu}$. We now analyze bounded differences. Since each edge can affect whether a subset $H$ is a good core for up to $\binom{n-2}{\nu-2}$ subsets of size $\nu$, we apply McDiarmid's inequality with bounded differences $\binom{n-2}{\nu-2}$.

    McDiarmid's inequality tells us that the probability that the number of good cores of size $\nu$ in $G$ is at most $\frac{2}{3} \binom{n}{\nu}$ is upper bounded by:
    $$\exp\left(- \frac{7^2 \binom{n}{\nu}^2}{30^2 \binom{n}{2} \cdot \binom{n-2}{\nu - 2}^2}\right) \leq \exp\left(-C n^2 / \nu^4 \right) \leq k^{-2n},$$
    where $C > 0$ is some constant, and the final inequality uses $\nu =O( k^9\log k)$.
\end{proof}

Next, we show that every core expands to cover almost all vertices, with failure probability decaying with the number of exceptional vertices:

\begin{lemma}\label{lem:coreExpand}
    Assume $k\le n^{1/3-\delta}$. There is $c_\sigma>0$ so that for every approximately balanced $A$ and $B\in \N$ and $\sigma \ge c_\sigma\log k$, $G \sim \cBnA$ is $(\sigma,B)$-linked with probability at least 
    \[1 -\exp(-\Omega(n\sigma/k))-\exp(-\Omega(n(B+1)/k)).\]
\end{lemma}
\begin{proof}
    We prove that it is linked with respect to the planted coloring $A$. 
    Fix an arbitrary collection $D=(D_1,\ldots,D_k)$ where $|D_i|\in [\sigma,O(k\log k)]$ and a set of $B+1$ possible bad vertices $V_B=(v_1,\ldots,v_{B+1})$. The number of choices is at most
    \[
    \binom{(1+0.01)n/k}{O(k \log k)}^k\cdot \binom{n}{B+1} \le n^{O(B+k^2\log k)},
    \]
    so it suffices to show that~\Cref{itm:core-inferred} and~\Cref{itm:s-inferred} hold with the claimed bound fixing a particular choice of $D,V_B$, since $n/k=\omega(B\log n+k^2\log k\log n)$. 

    First, fix an arbitrary $v\in A_i$. The event that $v$ is not strongly $D$-inferred is equivalent to having fewer than $0.01|D_j|$ edges to $D_j$ for some $j$. This event has probability $k\exp(-\Omega(\sigma))=\exp(-\Omega(\sigma))$, where the latter step uses $\sigma\ge c_\sigma\log k$ for a sufficiently large constant $c_\sigma$. Moreover, these events are independent for every $v,v'\in A_i$ and the partition is approximately balanced, so we obtain that item (1) holds with probability $1-\exp(-\Omega(\sigma n/k))$.

    Next, note that for every $v\in A_i$ and $S\subseteq A_j$ with $|S|/|A_j|\ge 0.9$ we have that
    \[\frac{|\Gamma(v)\cap A_j|}{|A_j|} \ge 0.49 \implies \frac{|\Gamma(v)\cap S|}{|S|} > 0.3\]
    so to prove that $v$ is strongly inferred by every such collection of sets (i.e., item (2) holds for $v$) it suffices to prove it has neighbor density $0.49$ in every other color class. 
    
    Next, let $A_j'$ be $A_j$ with all elements of $V_B$ deleted. 
    For $v\in V_B$ we have that $| \Gamma(v)\cap A_j' |$ is distributed as $\textsc{Bin}(|A_j'|,1/2)$ and hence by tail bounds (and since $A$ is balanced) the intersection is of size at least $(0.495)|A_j|\ge (0.49)|A_j'|$ with probability at least $1-\exp(-\Omega(n/k))$. By a union bound, this holds for all $j\neq i$ with approximately this probability. Moreover, since over the space $V_B\times \{A_1',\ldots,A_k'\}$ the events that the vertices are bad are independent (since they depend on disjoint sets of possible edges), we obtain that all are bad with probability $\exp(-\Omega(n(B+1)/k))$ as claimed.
\end{proof}

\subsection{Transferring between models}\label{sec:transfer-models}

We now transfer these results from $\cBnA$ to $\cUnk$. We require an easy counting fact:
\begin{fact}[\cite{dyer1989solution}]\label{fact:easy_counting_fact}
    Let $a_1,\ldots,a_k\in \N$ satisfy $\sum_i a_i=n$. Then
    \begin{displaymath}
        \sum_{i<j} a_ia_j = \binom{k}{2}(n/k)^2 - \frac{1}{2}\sum_i(a_i-n/k)^2.
    \end{displaymath}
\end{fact}

We can then prove the result. The remaining step is to transfer between models, from $G \sim \cBnA$ to $G \sim \cUnk$.
\begin{proof}[Proof of~\Cref{thm:good_property_probability}]
    For $A:[n]\ra[k]$, we say that $G$ is colorable by $A$, written $E_A(G)$, if there are no conflicting edges under $A$. Let ``$A \text{ bal}$'' be shorthand for $A$ being approximately balanced.
    
    For brevity, let ``$G \text{ bad}$'' be the event that $G$ is not $(\sigma,B)$-linked or does not have many cores. We have
    \begin{align*}
        \Pr_{G\sim \cUnk}[G\text{ bad}] &\le \sum_S\Pr_{G\sim \cUnk}[G\text{ bad} ~|~ E_A(G)]\Pr_{G\sim \cUnk}[E_A(G)]\\ 
        &\le \sum_{A:A\text{ bal}}\Pr_{G\sim \cUnk}[G\text{ bad} ~|~ E_A(G)]\Pr_{G\sim \cUnk}[E_A(G)] + \sum_{A:A\text{ not bal}}\Pr_{G\sim \cUnk}[E_A(G)]
    \end{align*}
    We will write $\rho_{\text{bal}} \defeq \sum_{A:\,A\text{ bal}} \Pr_{G\sim\cUnk}[E_A(G)]$ and
$\rho_{\text{not bal}} \defeq \sum_{A:\,A\text{ not bal}} \Pr_{G\sim\cUnk}[E_A(G)]$.
    
    Next, for a fixed partition $A$, the distribution $\cUnk$ conditioned on $E_A(G)$ is exactly $\cBnA$. This is true since the event $E_A(G)$ holds iff every edge of $G$ lies between two distinct classes of $A$. Such a graph is $k$-colorable (with $A$ as a proper coloring) and lies in the support of $\cUnk$ with equal probability. So $\cUnk$ conditioned on $E_A(G)$ is uniform over all graphs with edges in $\{(u,v):A(u)\neq A(v)\}$ and thus equals $\cBnA$. Therefore, the expression above equals:
    \begin{align*}
        &=\sum_{A:A\text{ bal}}\Pr_{G\sim \cBnA}[G\text{ bad}]\Pr_{G\sim \cUnk}[E_A(G)]+\rho_{\text{not bal}}\\
        &\le \max_{A:A\text{ bal}}\left\{\Pr_{G\sim \cBnA}[G\text{ bad}]\right\}\cdot \rho_{\text{bal}}+\rho_{\text{not bal}}\\
        &\le (\exp(-\Omega(n\sigma/k))+\exp(-\Omega(n(B+1)/k))+k^{-2n})\cdot \rho_{\text{bal}}+\rho_{\text{not bal}}
    \end{align*}
    where the final inequality follows from \Cref{lem:coreExpand} and \Cref{lem:manyGoodCores}. 

    Next, we use the fact that for a partition $A$ there are at least $2^{\sum_{i<j}a_ia_j}$ distinct graphs $G\in \cUnk$ with $E_A(G)$, and thus by considering an approximately balanced partition and applying \Cref{fact:easy_counting_fact} we have
    \[
    |\cUnk|\ge 2^{\binom{k}{2}(n/k)^2-k/2}.
    \]
    \paragraph{Bounding $\rho_{\text{bal}}$.}
    We have that $\rho_{bal} \le k!\cdot \E_{G}[NC(G)]$ where $NC(G)$ is the number of colorings (up to isomorphism) of $G$, which is bounded by $2$ by~\Cref{thm:expColorBounded}. Then since $n/k=\omega(\log k!)$ we have that this term is negligible. 

    \paragraph{Bounding $\rho_{\text{not bal}}$.}

    Fixing an arbitrary imbalanced partition $A$, we have $a_i \not \in (1 \pm \frac{1}{100k}) \cdot n/k$ for some $i$, so the number of graphs $G$ consistent with $A$ is at most 
    \[
    2^{\sum_{i<j}a_ia_j}\le 2^{\binom{k}{2}(n/k)^2-\frac{1}{2(100k)^2}(n/k)^2}
    \]
    and hence 
    \[
    \Pr_{G\sim \cUnk}[E_A(G)] \le 2^{-\Omega((n/k)^2)}.
    \]
    Finally, there are at most $k^n$ imbalanced partitions $A$, so
    \[
    \sum_{A:\text{$A$ not bal}}\Pr_{G\sim \cUnk}[E_A(G)]\le k^n\cdot 2^{-\Omega((n/k)^2)}\le k^{-4n}
    \]
    where the final inequality follows as $n^2/k^2=\omega(n\log k)$.
\end{proof}

\section{\ts{$k$}{k}-Coloring in Average \ts{$O(n k)$}{Linear} Time}\label{sec:good-n-good-k}
In this section, we analyze the randomized and deterministic algorithms.

We begin by presenting and describing the algorithmic subroutines (i.e., the ``phases'' described in \Cref{sec:alg-overview-coloring}) in \Cref{sec:algorithmic-subroutines}. We then give the algorithm and proofs of correctness and runtime bounds in \Cref{sec:the-algorithm}.

\newcommand{\InnerLoop}{\textsc{Inner-Loop}~}
\newcommand{\BuildGoodCore}{\textsc{Build-Good-Core}~}
\newcommand{\GoodCoreAdj}{\textsc{Good-Core-Neighbors}~}
\newcommand{\GoodCoreAdjNaive}{\textsc{Good-Core-Neighbors-Naive}~}
\newcommand{\ManyAdjRand}{\textsc{Many-Adjacencies-Random}~}
\newcommand{\ManyAdjDet}{\textsc{Many-Adjacencies-Deterministic}~}
\newcommand{\ManyAdjNaive}{\textsc{Many-Adjacencies-Naive}~}

\newcommand{\BadVertices}{\textsc{Bad-Vertices}}
\newcommand{\SampSize}{\textsc{Core-Color-Class-Size}}

\subsection{Algorithmic Subroutines}\label{sec:algorithmic-subroutines}

In this section, we give the algorithmic subroutines that we will leverage in our overall $k$-coloring algorithm. We explain how each subroutine fits into the overall framework and corresponds to the good graph properties proven in the previous section.

In the following algorithms, when an algorithm returns a ``coloring'' or ``partial coloring,'' this means that it returns $k$ disjoint subsets $C_1, C_2, \dots, C_k$ of the vertices. For this to be a \textit{full} coloring, $C_1 \cup C_2 \cup \dots \cup C_k = [n]$. The algorithms below will always return \textit{proper} colorings.

Our algorithm is constructed via a local approach; indeed, we are implicitly constructing a Local Computation Algorithm (LCA). See \Cref{sec:LCA-implementation} for the explicit LCA implementation of the algorithm.

Before presenting the algorithmic subroutines, we briefly recall the structure of the algorithm as we gave in the overview in \Cref{sec:overview}. On a $k$-colorable graph, the algorithm colors the graph in phases. It first samples a core and certifies that it is uniquely colorable. It then extends the core's coloring to most of the graph, coloring all but a small set of vertices. Finally, it colors the remaining small set of vertices by brute force over their few possible colorings. In the rare event that any of the earlier phases fail, the graph is colored directly using the $\tO(2^n)$-time algorithm of~\cite{coloringbruteforce}. A reader may find it useful to read the top-level algorithm in \Cref{sec:the-algorithm} before reading the algorithmic subroutines.

\subsubsection{Finding a good core}\label{sec:good-core-alg}

We begin with an algorithm that, given as input a potential core graph $H$, checks if a given $k$-coloring algorithm returns a coloring efficiently. It also certifies that $H$ is uniquely colorable, so that the coloring returned by the algorithm is known to be uniquely colorable. Certifying $H$'s unique colorability is accomplished via checking degree and codegree conditions of \Cref{def:goodCore} of a good core, which, if they hold, imply that the core is uniquely colorable, by \Cref{lem:core-unique-color}.

\begin{algorithm}[h]
\caption{Phase 1: Finding and $k$-coloring a good core}\label{alg:phase-1}
\KwIn{A core graph $H$, a $k$-coloring algorithm \textsc{Color-Graph}, and \SampSize}
\KwOut{Core sets $\{D_i\}_{i\in [k]}$ or \textsc{FAIL}}

\SetKwProg{Fn}{Procedure}{}{end}
\Fn{\BuildGoodCore}{
    Run \textsc{Color-Graph} on $H$ for up to $O(|H|^2/k)$ timesteps\;
    \If{after $O(|H|^2/k)$ time, \textsc{Color-Graph} returns \textsc{FAIL}}{
        \Return{\textsc{FAIL}}
    }
    \Else{
        \textsc{Color-Graph} outputs a $k$-coloring\;
        Let $C_1, C_2, \dots, C_k$ be the $k$ color classes of $H$ output by  \textsc{Color-Graph}\;
    }
    \If{$C_1, C_2, \dots, C_k$ is not a proper $k$-coloring of $H$}{
        \Return{\textsc{FAIL}}
    }
    Let $\varepsilon = 1/(10k)$\;
    \For{$i \in [k]$}{
    \If{$|C_i| <\SampSize$}{
        \Return{\textsc{FAIL}} 
    }
    \For{$u \in C_i$}{
    \If{$\exists j \neq i$ such that the degree of $u$ into $C_j$ is $\not \in (1 \pm \varepsilon^4) |H|/(2k)$}{
        \Return{\textsc{FAIL}} 
    }
    \For{$v \in C_i \setminus \{u\}$}{
    \If{$\exists j \neq i$ such that the number of common neighbors of $u$ and $v$ in $C_j$ is $\not \in (1 \pm \varepsilon^3) |H|/(4k)$}{
        \Return{\textsc{FAIL}} 
    }
    }
    }
    Let $D_i$ be the lexicographically first $\SampSize$ vertices in $C_i$\;
    }
    \Return{$D_1, D_2, \dots, D_k$}\;
}
\end{algorithm}

\medskip

\Cref{alg:phase-1} computes Phase 1 (Step 1) in the informal algorithm given in \Cref{sec:alg-overview-coloring}. This algorithm attempts to build a good core (\Cref{def:goodCore}), which satisfies: it has roughly balanced color classes, the algorithm \textsc{Kucera1995} of \Cref{thm:kuvcera} succeeds in producing a $k$-coloring on the core, and the degrees and codegrees are consistent. When we check these conditions, subgraphs that are accepted are uniquely colorable. Additionally, observe that the runtime of \BuildGoodCore is $O(|H|^3 / k)$. The size of $H$ will be chosen to be polynomial in $k$.

One part of \Cref{alg:phase-1} that we did not mention in \Cref{sec:alg-overview-coloring} is its \textit{subselection} 
component, where the algorithm returns the lexicographically first $\SampSize$ vertices in $C_i$, instead of all vertices in $C_i$, as the core's vertices in color class $i$. The algorithm initially considers a core of a larger size since our analysis needs a larger core to imply unique coloring of most potential cores. However, as this coloring is then propagated to the rest of the graph, the algorithm only needs to consider a smaller core. While subselection is not necessary for the runtime analysis of our randomized algorithm, we need this to obtain an $\widetilde{O}(nk)$ runtime in the case of the deterministic algorithm (without it, we would obtain an $O(n k^2)$ runtime since testing if $v$ is $D$-linked would run in time $O(k^2)$ per vertex).

\subsubsection{Expanding from a core}\label{sec:expanding-algo}

We first describe an algorithm that, given sets $D_1, D_2, \dots, D_k$ that are color classes of a good core, and given a vertex $u$ that is not in the core, checks adjacencies of $u$ to the core in order to color $u$ based on the unique coloring of the core.

\begin{algorithm}[H]
\caption{Phase 2, local: See if a vertex is adjacent to all but one color class of the good core}\label{alg:phase-2}
\KwIn{A $k$-colorable graph $G$, sets $D_1, D_2, \dots, D_k$, and a vertex $u \in V(G)$}
\KwOut{Color class $i \in [k]$, or \textsc{FAIL}}

\SetKwProg{Fn}{Procedure}{}{end}
\Fn{\textsc{Local-}\GoodCoreAdj}{
    \If{$u \in D_i$ for some $i \in [k]$}{\Return{$i$}}

    Let $\textsc{COLORS}=[k]$\;
        \For{$r\in [200k]$}{
        Sample $i\in \textsc{COLORS}$ and $v\in D_i$\;
        \lIf{$(u,v)\in G$}{Remove $i$ from $\textsc{COLORS}$}
        \lIf{$|\textsc{COLORS}|=1$}{
            Let $j$ be the remaining element. \Return{$j$}}
        }
    
    \Return{\textsc{FAIL}}\;
}
\end{algorithm}

\medskip
Observe that we will never have $|\textsc{COLORS}| = 0$ because we know that the graph is $k$-colorable. We may not know the color of $u$ from this procedure, but then we will return $\textsc{FAIL}$, not an incorrect coloring of $u$. We will never return an incorrect coloring of $u$ or find that it is not actually possible to color $u$ (we will only possibly find that we do not have enough information to color it yet).

We now give the global algorithm \GoodCoreAdj, which runs \textsc{Local}-\GoodCoreAdj on all vertices not in the core.

\medskip 

\begin{algorithm}[H]
\caption{Phase 2: Find vertices adjacent to all but one color class of the good core}\label{alg:phase-2-global}
\KwIn{A $k$-colorable graph $G$ and sets $D_1, D_2, \dots, D_k$}
\KwOut{Sets $S_1, S_2, \dots, S_k$, or \textsc{FAIL}}

\SetKwProg{Fn}{Procedure}{}{end}
\Fn{\GoodCoreAdj}{
    Initialize $S_1, S_2, \dots, S_k \leftarrow \emptyset$\;
    \ForEach{$u \in V(G)\setminus (D_1 \cup D_2 \cup \dots \cup D_k)$}{
        Let $j \leftarrow$ \textsc{Local-}\GoodCoreAdj$(G, D_1, D_2, \dots, D_k, u)$\;
        \lIf{$j \neq$ \textsc{FAIL}}{Add $u$ to $S_{j}$}
    }
    
    \Return{$S_1, S_2, \dots, S_k$}\;
}
\end{algorithm}

\medskip

We also give ``naive'' versions of these algorithms, which are less efficient but have a higher chance of success. These subroutines will be run with very low probability by our final algorithm.

\medskip

\begin{algorithm}[H]
\caption{Phase 2, local: See if a vertex is adjacent to all but one color class of the good core}\label{alg:phase-2-naive}
\KwIn{A $k$-colorable graph $G$, sets $D_1, D_2, \dots, D_k$, and a vertex $u \in V(G)$}
\KwOut{Color class $i \in [k]$, or \textsc{FAIL}}

\SetKwProg{Fn}{Procedure}{}{end}
\Fn{\textsc{Local-}\GoodCoreAdjNaive}{
    \If{$u \in D_i$ for some $i \in [k]$}{\Return{$i$}}
    \If{there exists $j \in [k]$ such that $u$ is adjacent to some vertex in each $D_i$ for every $i \neq j$}{
        \Return{$j$}
    }
    \Return{\textsc{FAIL}}\;
}
\end{algorithm}

\medskip

\begin{algorithm}[H]
\caption{Phase 2: Find vertices adjacent to all but one color class of the good core}\label{alg:phase-2-naive-global}
\KwIn{A $k$-colorable graph $G$ and sets $D_1, D_2, \dots, D_k$}
\KwOut{Sets $S_1, S_2, \dots, S_k$, or \textsc{FAIL}}

\SetKwProg{Fn}{Procedure}{}{end}
\Fn{\GoodCoreAdjNaive}{
    Initialize $S_1, S_2, \dots, S_k \leftarrow \emptyset$\;
    \ForEach{$u \in V(G)\setminus (D_1 \cup D_2 \cup \dots \cup D_k)$}{
        Let $j \leftarrow$ \textsc{Local-}\GoodCoreAdjNaive$(G, D_1, D_2, \dots, D_k, u)$\;
        \lIf{$j \neq$ \textsc{FAIL}}{Add $u$ to $S_{j}$}
    }
    \Return{$S_1, S_2, \dots, S_k$}\;
}
\end{algorithm}

\medskip

\Cref{alg:phase-2-global} corresponds to Phase 2 (Step 2) in the informal algorithm given in \Cref{sec:alg-overview-coloring}. This step colors vertices that have a $D$-inferred coloring (see \Cref{def:inferred-coloring}), for $D = \bigcup_{i \in [k]} D_i$. The algorithm \textsc{Good-Core-Neighbors} uses the property that many vertices in the graph have a \textit{strongly $D$-inferred coloring} in order to color a large fraction of vertices in the graph. If $D$ is \textit{uniquely} colorable, then vertices colored by this algorithm will also be uniquely colored. There is a small chance, however, this algorithm, or some other step starting with this core, will fail. In this case, if after enough iterations with a certain size of a core we have not found a core that succeeded, we do something more naive. We run \textsc{Good-Core-Neighbors-Naive}, which just uses that even more vertices have a \textit{$D$-inferred coloring} (not necessarily strongly inferred). The algorithm \textsc{Good-Core-Neighbors-Naive} checks \textit{all} adjacencies of a vertex $v$ to $D$ to try to color $v$.

\subsubsection{Sample neighbors for remaining vertices}\label{sec:sample-neighbors-algo}

Next, given large sets of vertices $S_1, S_2, \dots, S_k$ colored with each of the $k$ colors, and a vertex $u$ that is not already colored, the following algorithm attempts to color $u$ by sampling vertices in each of $S_1, S_2, \dots, S_k$ and checking adjacencies. We use the structural properties of \Cref{thm:good_property_probability} to argue that, with overwhelming probability, the remaining vertices can be efficiently colored in this way.

\medskip

\begin{algorithm}[H]
\caption{Phase 3 (randomized): See if a vertex is adjacent to many vertices colored in Phase 2 in all but one color class}\label{alg:phase-3-random}

\KwIn{A $k$-colorable graph $G$, sets $S_1, S_2, \dots, S_k$, and a vertex $u \in V(G)$}
\KwOut{A pair $(j, \textsc{Iter})$ with $j \in [k]$, or \textsc{FAIL}}

\SetKwProg{Fn}{Procedure}{}{end}
\Fn{\textsc{Local-}\ManyAdjRand}{
    \If{$u \in S_i$ for some $i \in [k]$}{\Return{$(i,0)$}}
    Let $\textsc{Iter}=0$\;
    Let \textsc{Possible-Colors} $\leftarrow \{1, 2, \dots, k\}$\;
        \While{$|\textsc{Possible-Colors}| > 1$ and \textsc{Iter} $\leq 100 n k$}{
            Sample $j\in \textsc{Possible-Colors}$ and $v\in S_j$\;
            \If{$u$ is adjacent to $v$}{
                Remove $j$ from \textsc{Possible-Colors}\;
            }
            Increment \textsc{Iter}\;
        }
        \If{$|\textsc{Possible-Colors}| = 1$}{
            Let $j$ be the remaining element of \textsc{Possible-Colors}\;
            \Return{$(j,\textsc{Iter})$}
        }
        \Else{
            \Return{\textsc{FAIL}}
        }
    }

\end{algorithm}

\medskip

We then turn this local algorithm into a global algorithm as follows.

\medskip

\begin{algorithm}[H]
\caption{Phase 3 (randomized): Find vertices adjacent to many vertices colored in Phase 2 in all but one color class}\label{alg:phase-3-random-global}

\KwIn{A $k$-colorable graph $G$ and sets $S_1, S_2, \dots, S_k$}
\KwOut{Sets $S_1', S_2', \dots, S_k'$, or \textsc{FAIL}}

\SetKwProg{Fn}{Procedure}{}{end}
\Fn{\ManyAdjRand}{
    Let $S_i'=\emptyset$ for every $i\in [k]$\;
    Let \textsc{Repeated} $=1$\;
    \ForEach{$u \in V(G)\setminus (S_1 \cup S_2 \cup \dots \cup S_k)$}{
        Let $o \leftarrow$ \textsc{Local-}\ManyAdjRand$(G, S_1, S_2, \dots, S_k, u)$\;
        \If{$o = (j,\textsc{Iter})$ for some $j \in [k]$}{
            Add $u$ to $S'_{j}$\;
            Increment \textsc{Repeated} by \textsc{Iter}\;
        }
        \lIf{\textsc{Repeated} $\ge 100nk$}{\Return{\textsc{FAIL}}}
    }
    \Return{$S_1', S_2', \dots, S_k'$}\;
    
}
\end{algorithm}

\medskip

\Cref{alg:phase-3-random-global} utilizes the fact that all but very few vertices in the graph have a strongly $S$-inferred coloring, where $S = (S_1, S_2, \dots, S_k)$ is a sufficiently large set with approximately balanced color classes. This allows Phase 3 to, with high probability, color most remaining uncolored vertices in the graph.

\medskip

\begin{algorithm}[H]
\caption{Phase 3 (naive): Find vertices adjacent to vertices colored in Phase 2 in all but one color class.}\label{alg:phase-3-random-uncareful}
\KwIn{A $k$-colorable graph $G$ and sets $S_1, S_2, \dots, S_k$}
\KwOut{Sets $S_1', S_2', \dots, S_k'$, or \textsc{FAIL}}

\SetKwProg{Fn}{Procedure}{}{end}
\Fn{\ManyAdjNaive}{
    Initialize $S_1', S_2', \dots, S_k' \leftarrow \emptyset$\;
    \ForEach{$u \in V(G)$}{
        Let \textsc{Possible-Colors} $\leftarrow \{1, 2, \dots, k\}$\;
        \ForEach{$v \in V(G)$}{
        \If{$u$ is adjacent to $v$}{
            Test if $v \in S_i$ for $i \in [k]$ \;
            \If{$v \in S_i$ for some $i \in [k]$}{
                Remove $i$ from \textsc{Possible-Colors}\;
            }
        }
        }
        \If{$|\textsc{Possible-Colors}| = 1$}{
            Let $j$ be the remaining element of \textsc{Possible-Colors}\;
            Add $u$ to $S_j'$\;
        }
        \Else{
            \Return{\textsc{FAIL}}
        }
    }
    \Return{$S_1', S_2', \dots, S_k'$}\;
    }
\end{algorithm}

\medskip

In \Cref{alg:phase-3-random-uncareful} (Phase 3, naive), we simply construct the sets $S_i'$ in brute-force fashion. This is because \Cref{alg:phase-3-random-uncareful} will be executed in a very unlikely case, and so we are not too concerned with polynomial-time factors. Instead, we need to find adjacencies if they exist, so we execute this less sophisticated algorithm.

\subsubsection{Color remaining vertices}

Our global algorithm will be fast on graphs that are $(\sigma, B)$-linked for some properly chosen parameters $\sigma$ and $B$. By definition of linked with these parameters, there will be at most $B$ vertices that may be uncolored after the previous phases have been run. We need to color these remaining vertices to obtain a proper $k$-coloring, which is what this step does.

However, \Cref{alg:phase-4} can fail to find a proper coloring given the existing color classes $F_1, F_2, \dots , F_k$. This can happen if Phase 1 starts with a \textit{bad} core. This core may not have been uniquely colorable, and so possibly the coloring was not able to be extended well to the rest of the graph. This can cause failures here, and when this happens, \textsc{Color-Remaining} returns FAIL. (The global algorithm will then try a different method.)

\begin{algorithm}[H]
\caption{Phase 4: Color remaining vertices}\label{alg:phase-4}
\KwIn{A $k$-colorable graph $G$, color classes $F_1, F_2, \dots, F_k$, and \BadVertices}
\KwOut{A proper $k$-coloring of $G$, or \textsc{FAIL}}

\SetKwProg{Fn}{Procedure}{}{end}
\Fn{\textsc{Color-Remaining}}{
    \If{there are more than \BadVertices ~vertices that are not in $F = F_1 \cup F_2 \cup \dots \cup F_k$ (i.e., not yet colored)}{
        \Return{\textsc{FAIL}}\;
    }
    \Else{
        Try all possible colorings of the at most \BadVertices ~remaining vertices\;
        \If{some coloring corresponds to a proper $k$-coloring of $G$}{
            \Return{the (lexicographically first) such coloring}\;
        }
        \Else{
            \Return{\textsc{FAIL}}\;
        }
    }
}
\end{algorithm}

\medskip 
If the global algorithm's more sophisticated method of coloring did not succeed on an input, we will instead throw out the coloring we have created so far and restart with an exhaustive search for a proper $k$-coloring. \Cref{alg:phase-5} captures this exhaustive search step, which will ensure that our algorithm outputs a proper $k$-coloring on every $k$-colorable graph. Here ``exhaustive search'' means invoking the algorithm of~\cite{coloringbruteforce}, which decides $k$-colorability and produces a coloring in time $\tO(2^n)$, independent of $k$. \Cref{alg:phase-5} will be run so infrequently that its $\tO(2^n)$ runtime will not negatively impact the average runtime. Getting this step to be run so infrequently that this does not impact the average runtime is a key technical challenge that we overcome in this paper.
\medskip

\begin{algorithm}[H]
\caption{Phase 5: Exhaustive search}\label{alg:phase-5}
\KwIn{A $k$-colorable graph $G$}
\KwOut{A proper $k$-coloring of $G$, or \textsc{FAIL}}

\SetKwProg{Fn}{Procedure}{}{end}
\Fn{\textsc{Exhaustive-Search}}{
    Compute a proper $k$-coloring of $G$ in time $\tO(2^n)$ using the algorithm of~\cite{coloringbruteforce}\;
    \Return{the resulting proper $k$-coloring of $G$}\;
}
\end{algorithm}

\subsection{The Randomized Algorithm}\label{sec:the-algorithm}
We now present our randomized algorithm (\Cref{alg:main-rand}) for $k$-coloring $k$-colorable graphs, whose average runtime over the set of all $k$-colorable graphs is $O(n \cdot k)$. In the pseudocode, whenever a subroutine returns \textsc{FAIL}, the current iteration is abandoned and the loop proceeds to the next iteration.

\begin{algorithm}[h]
\caption{$k$-Coloring a $k$-Colorable Graph (Optimal in $n, k$)}\label{alg:main-rand}
\KwIn{A $k$-colorable graph $G$}
\KwOut{A proper $k$-coloring of $G$}

\SetKwProg{Fn}{Procedure}{}{}
\Fn{\textsc{$k$-Coloring-Optimally-Randomized}}{
    Let $\nu$ be as in \Cref{def:manyGoodCores}\;
    \For{$\textsc{Iterations} = 0, 1, \dots, n^2 - 1$}{
        Sample $P\sim \binom{V}{\nu}$\;
        Run Algorithm \BuildGoodCore$(H=G_P,\text{algo}= \textsc{Kucera1995}, \SampSize = c_\sigma \log k)$\;
        Let $D_1, D_2, \dots, D_k$ be the outputs of \BuildGoodCore\unskip\;
        Run Algorithm \GoodCoreAdj$(G,\{D_1, D_2, \dots, D_k\})$\;
        Let sets $S_1, S_2, \dots, S_k$ be the outputs of \GoodCoreAdj\unskip\;
        Run Algorithm \ManyAdjRand$(G,\{S_1, S_2, \dots, S_k\})$\;
        Let sets $S_1', S_2', \dots, S_k'$ be the outputs of \ManyAdjRand\unskip\;
        
        \If{every vertex in $G$ is in some $S_i \cup S_i'$}{
            \Return{$C = \{S_1 \cup S_1', S_2 \cup S_2', \dots, S_k \cup S_k'\}$}\;
        }
    }

    \For{all $P\in \binom{V}{\nu}$}{
        Run Algorithm \BuildGoodCore$(H=G_P,\text{algo}= \textsc{Kucera1995}, \SampSize = c k \log k)$\;
        Let $D_1, D_2, \dots, D_k$ be the outputs of \BuildGoodCore\unskip\;
        Run Algorithm \GoodCoreAdjNaive$(G,\{D_1, D_2, \dots, D_k\})$\;
        Let sets $S_1, S_2, \dots, S_k$ be the outputs of \GoodCoreAdjNaive\unskip\;
        Run Algorithm \textsc{Many-Adjacencies-Naive}$(G,\{S_1, S_2, \dots, S_k\})$\;
        Let sets $S_1', S_2', \dots, S_k'$ be the outputs of \ManyAdjNaive\unskip\;

        Run Algorithm \textsc{Color-Remaining} on inputs: graph $G$, sets $\{F_i=D_i\cup S_i \cup S_i'\}$, and bad vertex bound \BadVertices$=ck\log k$\;
    
    \If{\textsc{Color-Remaining} returns a coloring $(C)$}{
        \Return{$C$}\;
    }
    }

    Run Algorithm \textsc{Exhaustive-Search} on the graph $G$\;
    \Return{the output of \textsc{Exhaustive-Search}}\;
}
\end{algorithm}

\medskip

This algorithm corresponds to the informal algorithm given in \Cref{sec:alg-overview-coloring}. The first outer for loop will succeed on awesome graphs (\Cref{def:awesomeokay}). The second outer for loop will succeed on okay graphs (\Cref{def:awesomeokay}). The exhaustive search will be run in the very rare case where a graph is neither awesome nor okay. We see that there is a close correspondence between the definitions of awesome and okay and the choice of parameters in the algorithm.

We call a single execution that samples a core, certifies and colors it, and attempts to
propagate the coloring an \textsc{Inner-Loop}. Both loops of \Cref{alg:main-rand} consist of repeated \textsc{Inner-Loop}s.

We now prove the main result.
\mainRand*
\begin{proof}
    We first argue correctness, then runtime. We assume that $k\le n^{1/36-\delta}$ for a positive constant $\delta>0$.
    \paragraph{Correctness.} In any iteration, \InnerLoop may succeed in $k$-coloring a $k$-colorable graph. We claim that \InnerLoop never returns an incorrect coloring, so this cannot harm us. First, by \Cref{lem:core-unique-color}, a good core is uniquely colorable, and hence the coloring of the core agrees with every global coloring of $G$. Given a unique coloring of the core, the partial coloring produced by \GoodCoreAdj (or \GoodCoreAdjNaive) and \ManyAdjRand (or \ManyAdjNaive) is unique. Finally, if \textsc{Color-Remaining} is run, then it properly $k$-colors the remainder of the graph given the unique partial coloring produced by the previous steps. Therefore, any iteration \InnerLoop that succeeds produces a proper $k$-coloring.
    
    If no iteration of \InnerLoop returns such a coloring, then \textsc{Exhaustive-Search} is executed, which by definition will ensure that a $k$-coloring of the graph is produced.

    \paragraph{Runtime.}
    We now prove a runtime bound. Let $A(G)$ be the event that $G$ is awesome and let $O(G)$ be the event that $G$ is okay. We have that $G\sim \cUnk$ is awesome with probability $1-\exp(-\Omega(n/k))$ and is okay with probability at least $1-3k^{-2n}$. 

    Let $T_A,T_O$ be the runtimes of the first and second loops.
    \begin{claim}
        $\E_G[T_A]=O(nk)$.
    \end{claim}
    \begin{proof}
        If $G$ is awesome, we claim each iteration successfully colors $G$ with probability at least $1/2$. First, we sample a good core $H$ with probability at least $2/3$. If $H$ is good, by \Cref{lem:core-unique-color} it is uniquely colorable, and hence every valid coloring $A$ of the entire graph $G$ agrees on $G_H$ (up to permuting the color classes). Thus, letting $A$ be the approximately balanced coloring on which $G$ is $(\sigma_A=c_\sigma\log k,0)$-linked, we have that the sets $D_i$ returned by step 1 (~\BuildGoodCore) satisfy $D_i\subseteq A_i$ and $|D_i|\ge \sigma_A$, and the sets $S_i$ returned by step 2 (~\GoodCoreAdj) satisfy $S_i\subseteq A_i$. 

        We next show that for every $i$ we recover sets $S_i\subseteq A_i$ that satisfy $|S_i|\ge 0.9|A_i|$ with overwhelming probability. By~\Cref{def:linked}, there is a set $L_i\subseteq A_i$, containing at least 99\% of $A_i$, whose vertices are strongly $D$-inferred. When the loop in \GoodCoreAdj reaches $v\in L_i$, the expected number of samples required to find connections to all but one $D_j$ is at most $200k$ (since we waste at most half our samples on our own color class in expectation), and so it is placed into $S_i$ with probability $0.99$. Moreover, these events are independent for every $v\in L_i$. Thus, via a simple concentration bound, we obtain that each $S_i$ is large enough with probability $1-\exp(-\Omega(n/k))\gg 0.99/k$.
        
        Hence, by~\Cref{itm:s-inferred} and the fact that $G$ is awesome, every vertex $v\in A_i$ is strongly $S$-inferred. Thus the expected number of iterations for \ManyAdjRand to color $v$ is at most $8k$, as attempting to color $v$ by its own color at most doubles the expected samples. Thus the expected time to color all vertices is at most $8nk$, and hence all vertices are colored before the iteration counter reaches $100nk$ with probability $0.9$, and hence all vertices are colored with the claimed probability.

        Moreover, each iteration takes time at most $\tO(k^{26})$ to color the core, $O(nk)$ to color vertices adjacent to the core, and $O(nk)$ to attempt to color the remaining vertices. Finally, $G$ is not awesome with probability $\exp(-\Omega(n/k))$, and the worst-case runtime of this phase is $\tO(n^2k^{26} +n^3k)$, so we are done.
    \end{proof}

    \begin{claim}
        $\E_G[T_O]=o(1)$.
    \end{claim}
    \begin{proof}
        First note that if $G$ is awesome, the probability that this phase executes at all is at most $2^{-n^2}$. Otherwise, the total runtime of the phase is at most $n^{\nu}\cdot \poly\left(n,2^{k\log^2 k}\right)$, and hence the expected runtime is at most 
        \[\Pr[\neg A(G)]\cdot n^{\nu}\cdot\poly\left(n,2^{k\log^2 k}\right)\le \exp(-\Omega(n/k))\cdot \poly(n^\nu)\] 
        which is negligible since $n/k=\omega(\nu\log n)$.
    \end{proof}

    Finally, let $T_{BF}$ be the runtime of the final brute-force stage. We have that the worst-case runtime is $\tO(2^n)$ by~\cite{coloringbruteforce}. We claim that no okay graph reaches this phase, which suffices to establish that $\E[T_{BF}]=o(1)$ by our definition of okay. Such a graph contains a good core $P$. Note that $|P|/2k\ge \sigma_O$ so each color class will be sufficiently large. Once the phase 2 loop reaches $P$ we will identify $S_1,\ldots,S_k$ where $|S_i|\ge 0.99|A_i|$ for the $A$ for which $G$ is $(\sigma_O,B_O)$-linked, and so all but $B$ vertices will be colored by \ManyAdjNaive, and so \textsc{Color-Remaining} will color the entire graph.
\end{proof}

\subsection{The Deterministic Algorithm}\label{sec:deterministic}

We define a new deterministic subroutine used by our deterministic algorithm.
Let $\cP=\{P_1,P_2,\ldots,P_{y_0}\}$ be a collection of vertex-disjoint sets of size $\nu$, where
$y_0 = O(k^2\log k\log n) = \tO(k^2)$ is fixed in the proof of \Cref{thm:mainDet}. Let $\cP_{\le i}$ denote $(P_1,\ldots,P_i)$, and let $\cL$ be an arbitrary set of $y_0\cdot k\cdot \log(n)$ vertices that does not intersect with $\cP$ (and this is possible since $y_0\cdot \nu +y_0k\log(n)=o(n)$) and let $\cL_{\le j}$ denote the first $j$ such vertices.

We give deterministic versions of the \ManyAdjRand subroutine:

\medskip

\begin{algorithm}[H]
\caption{Phase 3 (deterministic): Find vertices adjacent to many vertices colored in Phase 2 in all but one color class}\label{alg:phase-3-det}

\KwIn{A $k$-colorable graph $G$, sets $S_1, S_2, \dots, S_k$, and size $y$}
\KwOut{Sets $S_1', S_2', \dots, S_k'$, or \textsc{FAIL}}

\SetKwProg{Fn}{Procedure}{}{end}
\Fn{\ManyAdjDet}{
    Let $S_i'=\emptyset$ for every $i\in [k]$\;
    \ForEach{$u \in V(G)\setminus (S_1\cup\ldots\cup S_k)$}{
        Let \textsc{Possible-Colors} $\leftarrow \{1, 2, \dots, k\}$\;
        \ForEach{$v\in \cL_{\le \lceil y k \log n \rceil}$}{
        \If{$u$ is adjacent to $v$ and $v\in S_j$ for some $j\in[k]$}{
            Remove $j$ from \textsc{Possible-Colors}\;
        }
        }
        \If{$|\textsc{Possible-Colors}| = 1$}{
            Let $j$ be the remaining element of \textsc{Possible-Colors}\;
            Add $u$ to $S_j'$\;
        }
        \Else{
            \Return{\textsc{FAIL}}
        }
    }
    \Return{$S_1', S_2', \dots, S_k'$}\;
}
\end{algorithm}

\medskip

Recall that the algorithm \textsc{Many-Adjacencies-Random} used the property that all but very few vertices in the graph have a strongly $S$-inferred coloring for $S = (S_1, S_2, \dots, S_k)$. Therefore \textit{sampling} vertices in $S_i$ and checking if they are adjacent to $u$ sufficed. In our deterministic algorithm, of course, we do not want to \textit{sample} vertices. Instead, we look for adjacencies to an arbitrary set of vertices of a fixed size. Our correctness proof will show that this still suffices for our algorithmic purposes, at the cost of $\text{polylog}(n)$ factors in the average runtime of the deterministic algorithm.

Our deterministic algorithm is given in \Cref{alg:main-det}. As in the randomized algorithm, whenever a subroutine returns \textsc{FAIL}, the loop proceeds to the next iteration.

\begin{algorithm}[h]
\caption{$k$-Coloring a $k$-Colorable Graph (Optimal in $n, k$)}\label{alg:main-det}
\KwIn{A $k$-colorable graph $G$}
\KwOut{A proper $k$-coloring of $G$}

\SetKwProg{Fn}{Procedure}{}{}
\Fn{\textsc{$k$-Coloring-Deterministic}}{
    \For{$y\in\{1,\ldots,y_0\}$}{
    \For{$P \in \cP_{\leq y}$}{
        Run Algorithm \BuildGoodCore ($H = G_{P}$, algo $=$ \textsc{Kucera1995}, $\SampSize=c_{\sigma}\log k$)\;
        Let $D_1, D_2, \dots, D_k$ be the outputs of \BuildGoodCore\unskip\;
        Run Algorithm \GoodCoreAdjNaive$(G,\{D_1, D_2, \dots, D_k\})$\;
        Let sets $S_1, S_2, \dots, S_k$ be the outputs of \GoodCoreAdjNaive\unskip\;

        Run Algorithm \ManyAdjDet  ($G$, $\{S_1, S_2, \dots, S_k$\}, $y$)\;
        Let sets $S_1', S_2', \dots, S_k'$ be the outputs of \ManyAdjDet\unskip\;
        
        \lIf{every vertex in $G$ is in some $S_i \cup S_i'$}{
            \Return{$C = \{S_1 \cup S_1', S_2 \cup S_2', \dots, S_k \cup S_k'\}$}
        }
        }}

    \For{$P\in \binom{V}{\nu}$}{
        Run Algorithm \BuildGoodCore ($H = G_{P}$, algo $=$ \textsc{Kucera1995},  $\SampSize=ck\log k$)\;
        Let $D_1, D_2, \dots, D_k$ be the outputs of \BuildGoodCore\unskip\;
        Run Algorithm \GoodCoreAdjNaive$(G,\{D_1, D_2, \dots, D_k\})$\;
        Let sets $S_1, S_2, \dots, S_k$ be the outputs of \GoodCoreAdjNaive\unskip\;
        Run Algorithm \ManyAdjNaive$(G,\{S_1, S_2, \dots, S_k\})$\;
        Let sets $S_1', S_2', \dots, S_k'$ be the outputs of \ManyAdjNaive\unskip\;
        Run Algorithm \textsc{Color-Remaining} on inputs: graph $G$, sets $\{F_i=D_i\cup S_i\cup S_i'\}$, and bad vertex bound $\BadVertices=ck\log k$\;
        \lIf{\textsc{Color-Remaining} returns a coloring ($C$)}{
            \Return{$C$}
        }
    }

    Run Algorithm \textsc{Exhaustive-Search} on the graph $G$\;
    \Return{the output of \textsc{Exhaustive-Search}}\;
}
\end{algorithm}

\medskip

Let us observe what changed compared to the randomized algorithm. Whenever we ran \textsc{Many-Adjacencies-Random} before, we now run \textsc{Many-Adjacencies-Deterministic}. Additionally, our way of choosing $P$ on which we try to find a good core changes, as we will need a more refined way of choosing $P$ to achieve the $\widetilde{O}(nk)$ runtime in the deterministic setting, and we also can no longer \textit{sample} sets $P$.

We now analyze the average runtime of the deterministic algorithm. In proving the average runtime bound of the deterministic algorithm, the most important observation we use is that the distribution $\cUnk$ of $k$-colorable graphs (and all subclasses we use for analysis, such as graphs with many good cores) is vertex-symmetric, of which we make essential use in our derandomization: 
\begin{observation}\label{obs:invariant} 
    For every vertex permutation $\pi$ and vertex-symmetric distribution $\cD$, we have that $\cD$ and $\pi(\cD)$ are equal as distributions. Consequently, for every distribution over permutations $\Xi$, $\pi(G)$ where $\pi \sim \Xi$ and $G\sim \cD$ and $G$ where $G\sim \cD$ are equal as distributions.
\end{observation}
\begin{observation}
     The set of $k$-colorable graphs, the set of awesome graphs, and the set of okay graphs (\Cref{def:awesomeokay}) are all vertex-symmetric.
\end{observation}

Let $\Pi$ be the uniform distribution on permutations $\pi:[n]\ra[n]$. Let $\Pi_{L}$ be the uniform distribution on automorphisms of $L$ (i.e., permutations that fix every element of $L$).
We use the following model of random generation of permutations that fix certain vertices:
\begin{observation}\label{obs:genperm}
    For every $L\subseteq [n]$ and automorphism $\tau$ of $L$, the following procedure generates a random automorphism $\pi$ of $L$. Initialize $A_0=\emptyset$. For $i=1,2,\ldots$ ranging over the elements of $[n]\setminus L$ in increasing order, draw a random $a\in [n]\setminus (A_{i-1}\cup L)$, set $\pi(a)=i$, and set $A_i=A_{i-1}\cup \{a\}$. Finally, output $\pi'=\tau\circ \pi$.
\end{observation}

\newcommand{\cA}{\mathcal{A}}
 In our deterministic algorithm, we first attempt to find a good core in $\cP$, and otherwise brute force over all sets of size $\nu$ in some arbitrary order. Since this latter stage will take time $\poly(n)\cdot n^{\nu}$, we wish to show that with high probability over $G\sim \cUnk$ there is $P\in \cP$ where $G_{P}$ is a good core. We do so for the set of awesome graphs, whose uniform distribution we denote by $\cA$.

\begin{lemma}\label{lem:derand:earlyCore}
    Let $G\sim \cA$. The probability that there exists $P\in \cP_{\le t}$ such that $G_{P}$ is a good core is at least $1-\exp(-\Omega(t))$.
\end{lemma}
\begin{proof}
    \newcommand{\cG}{\mathcal{G}}
    By~\Cref{obs:invariant}, it suffices to prove that for an arbitrary awesome graph $G$ this
    event holds for $\pi(G)$ over a random permutation $\pi\sim \Pi$ with the claimed probability. We generate $\pi$ in $t$ stages as follows. We first initialize $Q=\emptyset$, then at stage $i$ we select a set $Q_i\subseteq [n]\setminus Q$ of $\nu$ random vertices, set $\pi(Q_i)=P_i$ (which maps the freshly chosen vertices onto the $i$-th fixed core $P_i\in\cP$), and set $Q=Q\cup Q_i$. This procedure generates a random permutation by~\Cref{obs:genperm}.
    
    Let $\cG\subseteq \binom{[n]}{\nu}$ be the set of good cores present in $G$. By definition of awesome we have $|\cG|\ge (2/3)\binom{n}{\nu}$. We prove that for an arbitrary previously fixed set $Q$ and stage $i$ we have $Q_i\in \cG$ with probability at least $(2/3)-(0.01)$, which suffices to establish the result. This follows as the number of good cores incident to an arbitrary vertex $v$ is at most $\binom{n}{\nu-1}$, so the number of remaining good cores not incident to any vertex in $Q$ is at least
    \[
    |\cG|-|Q|\cdot \binom{n}{\nu-1}\ge (0.999)|\cG|. \qedhere
    \]
\end{proof}

Next, we prove that we can find paths to color all vertices with high probability by simply examining the first $O(\log n)$ vertices outside the core.
\begin{lemma}\label{lem:derand:route}
    Let $\cA_{Q}$ be the uniform distribution over the set of awesome graphs conditioned on a fixed, arbitrary assignment $Q$ to the subgraph $\cP_{\le t}$.
    If $P \subseteq Q$ is a good core, with probability at least $1-nk\cdot 2^{-\Omega(\ell/k)}$ over $G\sim \cA_{Q}$ the following occurs. Let $D_1, D_2, \dots, D_k$ be defined as before in terms of $P$ (as the outputs of \BuildGoodCore on $P$). Then for every $v\in V\setminus P$, for all but one $j\in [k]$ there are $u_j\in \cL_{\le \ell}\cap \Gamma(v)$ such that $u_j$ is $D$-inferred and colored $j$.
\end{lemma}
\begin{proof}
    It suffices to prove the bound $1-k\cdot 2^{-\Omega(\ell/k)}$ for a single arbitrary $v\notin P$, so fix such a $v$ and $Q$ such that $P\subseteq Q$ and $G_P$ is a good core.
    
    Next, the distribution $\cA_{Q}$ is invariant under permutations $\Pi'$ that fix $\cP_{\le t},v$, so it suffices to prove the bound holds for $\pi(G)$ where $\pi\sim \Pi'$ and $G\in \cA_{Q}$ is arbitrary. Fix such a $G$. 
    
    Now let $D_1,\ldots,D_k$ be defined as before in terms of $P$. Since $P$ is a good core it is uniquely colorable, and hence we have $D_i\subseteq A_i$ for the coloring $(A_1,\ldots,A_k)$ for which $G$ is $(\sigma_A,0)$-linked, and this coloring is approximately balanced. Thus, let $S_1,\ldots,S_k$ be the sets with $S_i\subseteq A_i$ that are $D$-inferred, and because $G$ is awesome we have $|S_i|\ge 0.99|A_i| =\Omega(n/k)$.

    WLOG $v\in A_i$. For $j\neq i$, let $R_j = S_j\cap \Gamma(v)\cap \cP_{\le t}^c\cap \{v\}^c$. Observe that no element of $R_j$ is fixed by $\Pi'$, so it suffices to prove that some element of $R_j$ is mapped to $\cL$ by $\pi$ for every $j$ with the desired probability.

    Fix a single such $j$. We have that
    \begin{align*}
        |R_j| &\ge |S_j\cap \Gamma(v)| - t\nu\\
        &\ge |S_j|(0.4) - t\nu\\
        &= \Omega(n/k)
    \end{align*}
    where the first line uses the bound on the core size, the second line uses that $G$ is awesome, and the third line uses that $|S_j|=\Omega(n/k)$ and $t\nu=o(n/k)$.
    
    The probability that $\pi$ maps no elements of $R_j$ to $\cL_{\le \ell}$ is therefore at most $(1-\Omega(n/k)/n)^\ell=\exp(-\Omega(\ell/k))$ as claimed, where we generate $\pi$ as in \Cref{obs:genperm}.
\end{proof}

We can then prove the deterministic algorithm has the desired runtime.

\mainDet*
\begin{proof}
    We assume $k\le n^{1/36-\delta}$ for a positive constant $\delta>0$. Correctness is immediate from the same analysis as~\Cref{thm:mainRand}, since we only report a coloring if it is valid and otherwise use brute-force to find a valid coloring.

    \paragraph{Bounding the runtime on awesome graphs.}
    We call an awesome graph $y$-fast if some $P\in \cP_{\le y}$ is a good core and such a core colors all vertices via $\cL_{\le y\cdot k\cdot \log n}$.

    \begin{claim}\label{clm:yfast}
        $G\sim \cA$ is $y$-fast with probability $1-\exp(-\Omega(y))$.
    \end{claim}
    \begin{proof}
        Let $E$ be the event that $G$ is $y$-fast. We have
        \begin{align*}
            \Pr_{G\sim \cA}[E] &= \sum_{Q:P\subseteq Q \text{ is good}}\Pr_{G\sim \cA_{Q}}[E|G_{\cP_{\le y}}=Q]\Pr_{G\sim\cA}[G_{\cP_{\le y}}=Q]\\
            &\ge \min_{Q:P\subseteq Q \text{ is good}}\{\Pr_{G\sim \cA_{Q}}[E|G_{\cP_{\le y}}=Q]\}\cdot \sum_{Q:P\subseteq Q \text{ is good}}\Pr_{G\sim\cA}[G_{\cP_{\le y}}=Q]\\
            &=(1-nk\cdot \exp(-\Omega(y \log n)))(1-\exp(-\Omega(y))) = 1-\exp(-\Omega(y))
        \end{align*}
        where the final line follows from~\Cref{lem:derand:earlyCore} and~\Cref{lem:derand:route}.
    \end{proof}
    The following claim is straightforward.
    \begin{claim}
        If $G$ is $y$-fast, it is colored in stage $y$ of the loop. Moreover, stage $y$ has runtime $\tO(y\nu^3/k+y^2kn)$.
    \end{claim}
    Finally, note that every awesome graph is colored in worst-case time $W=n^{\nu}\cdot \poly(n,2^{k^2\log k})$ (in particular, it is okay and we show the same statement for okay graphs below). We choose $t,\ell$ and $y_0=O(k^2\log k\log n)$ such that the failure probability in~\Cref{clm:yfast} is at most $1/W$ (equivalently, the success probability is at least $1-1/W$).

    We can then bound the average runtime over awesome graphs. Let $E_y(G)$ be the event that $G$ is $y$-fast, and note that we only reach stage $y$ of the loop on $G$ if $\neg E_{y-1}(G)$ holds. We have:
    \begin{align*}
        \E_{G\sim \cA}[T(G)] &\le \sum_{y\in [y_0]}\tO(y\nu^3/k+y^2kn)\cdot \Pr_{G\sim \cA}[\neg E_{y-1}(G)]+\Pr_{G\sim \cA}[\neg E_{y_0}(G)]\cdot W\\
        &\le \tO(nk+\nu^3/k)+\sum_{y\in [y_0]} \tO(y\nu^3/k+y^2kn)\cdot \exp(-\Omega(y)) +1\\
        &= \tO(nk+\nu^3/k).
    \end{align*}
    
    \paragraph{Bounding the runtime on not-awesome graphs.}
    By an identical argument to the randomized algorithm, every okay graph is colored in time $n^{\nu}\cdot \poly(n,2^{k\log^2 k})$. Thus, the contribution to the overall expectation from non-awesome but okay graphs is $o(1)$ by~\Cref{thm:good_property_probability}. Because all okay graphs are colored in this way, only a $3k^{-2n}$ fraction of graphs reach the brute-force phase, so the contribution to the overall expectation from not-okay graphs is $o(1)$.
\end{proof}

\section{Lower Bound of \ts{$\Omega(n k)$}{nk}}\label{sec:LB}

We now turn to analyzing the runtime that any algorithm must have on average. We prove the following stronger statement.

\begin{theorem}\label{thm:lower-bound}
    Fix $k\le n^{1/3}$. Any (zero-error randomized or deterministic)\footnote{As in our upper bound, we require worst-case correctness, in that over every random string $r$ and graph $G$ the algorithm never returns an incorrect coloring.} algorithm that correctly $k$-colors every $k$-colorable graph must take time $\Omega(n k)$ on every input.\footnote{Our algorithms require access only to the adjacency matrix of the graph. Our lower bound holds in the stronger model where we allow access to an (unsorted) adjacency list for each vertex.} 
\end{theorem}

Since we are considering algorithms whose runtime is measured on average but that must be correct in the worst case, this implies an $\Omega(n \cdot k)$ average  runtime lower bound. 
\begin{proof}[Proof of \Cref{thm:lower-bound}]
    Suppose there is an input $G=(V,E)$ such that the algorithm (with random string $r$) halts after time $nk/c$, for a constant $c>1$ to be chosen later, and returns a coloring $C:V\ra [k]$. We prove that this algorithm must color a $k$-colorable graph $G'$ incorrectly. It is clear that the algorithm queries the adjacency matrix of $G$ at most $nk/c$ times; let $Q\subseteq V\times V$ denote the queried entries of the matrix. By an averaging argument (and taking $c$ sufficiently large), at least $2/3$ of the vertices are incident to fewer than $k-2$ queried entries. We denote these deficient vertices $M\subseteq V$. Let 
    \[M_i = M \cap C^{-1}(i)
    \]
    be the subset of deficient vertices the algorithm colors $i$.

    \begin{claim}
        There exists $i\in [k]$ and $v,v'\in M_i$ such that $(v,v')\notin Q$.
    \end{claim}
    \begin{proof}
        We have $|M|\ge 2n/3$, so there is some $i$ such that $|M_i|\ge 2n/3k$. Then the number of potential pairs is 
        \[
        \binom{|M_i|}{2} \ge (2n/3k)^2/4 > |Q|
        \]
        where the final inequality follows from choosing $c$ large enough, so some pair must be unqueried.
    \end{proof}
    We fix this pair $v,v'$ and use it to falsify the algorithm. In what follows, for a graph $G$ and queried entries $Q$, let $G_Q$ be the subgraph on the queried entries.
    \begin{claim}
        There exists another $k$-colorable graph $G'$ such that $G'_Q=G_Q$ and yet in every valid coloring $C'$ of $G'$, we have $C'(v)\neq C'(v')$.
    \end{claim}
    
    \begin{proof}
        First consider the graph $G''$ that contains only the edges present in $G_Q$. It is clear that $G''$ is $k$-colorable (since we only delete edges from the $k$-colorable graph $G$). Moreover, we claim there is a $k$-coloring of $G''$ where $v,v'$ are colored differently. This is because $d_{G''}(v)\le k-2$ and $d_{G''}(v')\le k-2$, so for any valid partial coloring $P:V\setminus \{v,v'\}$ of $G''$ there are at least $2$ colors that are not assigned to any of their neighbors (and so we can extend $P$ to a full coloring $P'$ assigning different colors to them).
        Then let $G'$ be this graph with edge $(v,v')$ added. This graph is still $k$-colorable since $P'$ remains a valid coloring, and obviously any coloring of $G'$ must assign $v$ and $v'$ different colors since they are neighbors.
    \end{proof}
    But the algorithm (on the same random string $r$) on this graph $G'$ will act identically and hence assign the same color to $v$ and $v'$, which violates the worst-case correctness guarantee.
\end{proof}

\bibliography{bibliography}

\appendix

\section{Local Computation Algorithm Implementation}\label{sec:LCA-implementation}

In this section, we give a local computation algorithm (LCA) implementation of our randomized algorithm (\Cref{alg:main-rand}). At a high level, LCAs give local access to a global, consistent solution to a problem -- in this case, a $k$-coloring. We begin with the definitions of LCAs and average-case LCAs. Then, we show how to implement our randomized algorithm as an average-case LCA.

\subsection{Definitions}

\begin{definition}[Local computation algorithm \cite{DBLP:conf/innovations/RubinfeldTVX11, DBLP:conf/soda/AlonRVX12}]
    A local computation algorithm (LCA) $\mathcal{A}$ for a problem $\Pi$ is an oracle satisfying the following properties. $\mathcal{A}$ has probe access to the input (call it $G$), a sequence of random bits ($\vec{r}$), and local memory. Given any admissible query $q$ to the output of problem $\Pi$ on $G$, $\mathcal{A}$ is tasked with using only oracle access (``probes'') to $G$, random bits from $\vec{r}$, and its local memory to answer $q$. After giving the answer to any query, the LCA $\mathcal{A}$ erases its local memory, including $q$ and the response. All of the responses that $\mathcal{A}$ gives to queries must be consistent with one global, valid solution $X$ to the problem $\Pi$ on $G$.
\end{definition}

We now set up the definition of the ``probe complexity'' of $\mathcal{A}$. The probe complexity of an LCA is defined as the maximum over all inputs $G$ and maximum over all queries to $G$ of the expected number of probes, over the choice $\vec{r}$ of randomness. More formally:

\begin{definition}[Probe complexity of an LCA]
   Consider an LCA $\mathcal{A}$ which has probe access to input $G$. First, on any query $q$, let $T_{\mathcal{A}}(G, q)$ be the expected (over the choice of $\vec{r}$) number of probes $\mathcal{A}$ must make to answer $q$ on $G$. Let $T_{\mathcal{A}}(G) = \max_q T_{\mathcal{A}}(G, q)$.

   The LCA $\mathcal{A}$ has probe complexity $T(n)$ if the maximum value of $T_{\mathcal{A}}(G)$ over any input $G$ of size $n$ is $T(n)$.
\end{definition}

The paper \cite{beyond_worst_case_LCA} defined average-case local computation algorithms as follows.

\begin{definition}[Average-case local computation algorithm \cite{beyond_worst_case_LCA}]
    An oracle $\mathcal{A}$ is an average-case local computation algorithm for problem $\Pi$ and a distribution $\mathcal{G}$ over size-$n$ objects if, with probability at least $1 - \frac{1}{n}$ for $G \sim \mathcal{G}$, $\mathcal{A}$ with probe access to $G$ is an LCA for $\Pi$ on $G$.

    An average-case LCA $\mathcal{A}$ is said to have average-case probe complexity $T(n)$ if $T(n) = \mathbb{E}_{G \sim \mathcal{G}}[T_{\mathcal{A}}(G)].$ It has average-case time $t(n)$ if the expected (over $G \sim \mathcal{G}$) average time (over queries $u$) per query is $t(n)$.
\end{definition}

The definition of average-case probe complexity takes the average over all inputs of the maximum number of expected (over $\vec{r}$) probes over all queries. This definition is much stronger than the interpretation of ``average-case'' that would take a maximum over queries of the average probe complexity over all graphs.

We consider LCAs that \textit{never} give access to incorrect solutions to the problem $\Pi$ on input $G$. For our problem, this corresponds to never giving access to a bad or improper partial coloring. On each vertex, the algorithm either returns a coloring corresponding to a partial proper coloring, or returns \textsc{FAIL}. We define this as follows:

\begin{definition}[Zero-error average-case local computation algorithm]
    An oracle $\mathcal{A}$ is a zero-error average-case local computation algorithm for problem $\Pi$ and a distribution $\mathcal{G}$ over size-$n$ objects if it is an average-case LCA for $\Pi$ and $\mathcal{G}$ and the algorithm \textit{never} gives access to an incorrect/improper output for $\Pi$. It always gives access to at least a correct partial output and returns \textsc{FAIL} on the remaining vertices.
\end{definition}

Since a zero-error average-case LCA is an average-case LCA, it must give access to a full correct output for $\Pi$ with probability at least $1 - \frac{1}{n}$ for $G \sim \mathcal{G}$.

\subsection{Our Results}

In this section, we give our results regarding average-case LCAs, the LCA itself, and the proof of its correctness and probe complexity.

\begin{theorem}\label{thm:LCA}
    Assume $2 \le k \le n^{1/37}$. There is a zero-error average-case local computation algorithm for the uniform distribution over $k$-colorable graphs with average-case probe complexity $O(k^9 \log k)$ and average-case time $O(k^{26} \mathrm{polylog} (k))$.
\end{theorem}

\paragraph{Relation to worst-case LCAs.} LCAs for $(\Delta + 1)$-coloring 
have been previously considered \cite{DBLP:conf/innovations/RubinfeldTVX11, DBLP:conf/podc/ChangFGUZ19, DBLP:journals/corr/abs-2103-10990, DBLP:conf/icalp/DorobiszK23}. 
However,  even for 2-coloring, there are no
worst-case $o(n)$-time LCAs for $k$-coloring: 
Consider the graph consisting of two complete bipartite graphs 
connected by a single edge. 
One can use standard techniques to design a constant-time algorithm that outputs a consistent coloring for each complete bipartite graph. 
However, until we find the edge connecting the two components, 
the algorithm cannot know how to put the two colorings together. 
Nevertheless, average-case LCAs 
can be made significantly faster.

\paragraph{Algorithm overview.}
The average-case LCA implementation utilizes two algorithmic subroutines introduced in \Cref{sec:algorithmic-subroutines}. First, \BuildGoodCore colors a potential core subgraph, and certifies its unique colorability. \textsc{Local}-\GoodCoreAdj checks adjacencies to the core subgraph to color a vertex via the unique coloring of the core subgraph. 

On input $u$, the LCA first checks if $u$ is in the core,
and if so, colors it accordingly.  If not, the  LCA checks if \textsc{Local}-\GoodCoreAdj 
successfully colors $u$.   Otherwise, the LCA samples random vertices $v$, checks if they are adjacent to $u$, and checks if \textsc{Local}-\GoodCoreAdj successfully colors $v$. If the LCA finds neighbors of $u$ of every color but one, the LCA can then color $u$ with the final color.

We now give the average-case local computation algorithm. 
\medskip

\begin{algorithm}[H]
\caption{$k$-Coloring Local Computation Algorithm}\label{alg:lca}
\KwIn{A $k$-colorable graph $G$, a vertex $u \in V(G)$, and a shared random tape $\vec{r}$}
\KwOut{A consistent coloring $j$ of $u$}

\SetKwProg{Fn}{Procedure}{}{}
\Fn{\textsc{$k$-Coloring-LCA}}{
    Let $\nu$ be as in \Cref{def:manyGoodCores}\;
    Let $D_1,\dots,D_k \leftarrow \textsc{FAIL}$\;
    \While{$D_1,\dots,D_k = \textsc{FAIL}$ and the number of $P \in \binom{V}{\nu}$ sampled is $\exp(o(n/k))$}{
        Sample $P\sim \binom{V}{\nu}$ according to $\vec{r}$\;
        Run \BuildGoodCore$(H=G_P,\text{algo}=\textsc{Kucera1995}, \SampSize = c_\sigma \log k)$\;
        \lIf{it does not return \textsc{FAIL}}{Let $D_1,\dots,D_k$ be its outputs}
    }
    \lIf{$D_1,\dots,D_k = \textsc{FAIL}$}{\Return{\textsc{FAIL}}}

    Let $o \leftarrow$ \textsc{Local-}\GoodCoreAdj$(G,\{D_1,\dots,D_k\}, u)$\;
    \lIf{$o \neq \textsc{FAIL}$}{\Return{$o$}}

    Initialize \textsc{Possible-Colors} $=[k]$\;
    \While{$|\textsc{Possible-Colors}| \geq 2$ and not all vertices have been sampled}{
        Sample a vertex $v \in V(G)$\;
        \If{$v$ is adjacent to $u$}{
            Let $o \leftarrow$ \textsc{Local-}\GoodCoreAdj$(G,\{D_1,\dots,D_k\}, v)$\;
            \lIf{$o \neq \textsc{FAIL}$}{Remove $o$ from \textsc{Possible-Colors}}
        }
    }
    \If{$|\textsc{Possible-Colors}| = 1$}{\Return{the remaining color in \textsc{Possible-Colors}}}
    \Else{\Return{\textsc{FAIL}}}
}
    
\end{algorithm}

\begin{proof}[Proof of \Cref{thm:LCA}] 

    We first argue correctness, then runtime. We assume that $k\le n^{1/36-\delta}$ for a positive constant $\delta>0$.
    
    \paragraph{Correctness.} We prove that the $k$-\textsc{Coloring-LCA} algorithm provides access to a consistent global proper $k$-coloring of the graph with high probability. We also prove that the algorithm never outputs an incorrect $k$-coloring; if it does not output a full proper coloring, the algorithm returns \textsc{FAIL} on some vertices.

    First, observe that, if any good core (\Cref{def:goodCore}) exists in the graph, the LCA on any vertex $u$ will find the same good core $P$, since $P$ is found according to the shared random tape $\vec{r}$. A good core is certifiably uniquely colorable (as we show in \Cref{lem:core-unique-color}), and we color it with the algorithm \textsc{Kucera1995} (see \Cref{thm:kuvcera}; \cite{kuvcera1995expected}). On every query, the LCA uses the same deterministic coloring of the good core, as the algorithm \textsc{Kucera1995} is deterministic.
    
    In the remaining steps of the algorithm, we only color vertices if their coloring is \textit{unique} given the \textit{unique coloring of the good core}. Therefore, since each vertex finds its color uniquely based on the same coloring of the same good core, and we have assumed that the input graph is $k$-colorable, the algorithm always outputs at least a partial proper $k$-coloring of $G$. It is thus zero-error.

    The algorithm outputs a full proper $k$-coloring of $G$ when no vertices output \textsc{FAIL}. We prove that this occurs with probability at least $1 - \exp(-\Omega(n/k))$. First, by \Cref{thm:good_property_probability} (structural theorem), for every $B \in \mathbb{N}$ and $\sigma \geq c_{\sigma} \log k$, with probability at least $1-\exp(-\Omega(n\sigma/k))-\exp(-\Omega(n(B+1)/k))-2k^{-2n}$, $G \sim \mathcal{U}$ is $(\sigma, B)$-linked and has many good cores. 
    
    Recall the definition of strongly $H$-inferred coloring from \Cref{def:inferred-coloring}: Let $H$ be a properly $k$-colored subgraph of a graph $G$, with color classes $\{H_1, H_2, \dots, H_k\}$. We say that a vertex $v \in G$ has a \emphdef{strongly $H$-inferred coloring} if $|\Gamma(v) \cap H_i|/|H_i| \geq 0.01$ for all but one $i \in [k]$.
    
    By applying \Cref{def:linked} with $B = 0$, \Cref{thm:good_property_probability} tells us that with probability $1 - \exp(-\Omega(n/k))$ there is an approximately balanced valid coloring $(A_1, A_2, \dots, A_k)$ such that, for every collection of sets $D_i \subseteq A_i$ with $|D_i| \in [\sigma, O(k \log k)]$, we have:
    \begin{enumerate}
        \item For each $i \in [k]$, at least a 0.99 fraction of the vertices $v \in A_i$ are strongly $D$-inferred for $D = \{D_1, D_2, \dots, D_k\}$, and thus will be colored by \textsc{Local-}\GoodCoreAdj$(G, D, v)$.
        \item Suppose at least a 0.9 fraction of vertices $v \in A_i$ for each $i \in [k]$ is colored by \textsc{Local-}\GoodCoreAdj; we call this set $T_i$ for each $i \in [k]$. In this case, any remaining $u \not \in \cup_{i \in [k]} T_i$ is strongly $\{T_1, T_2, \dots, T_k\}$-inferred. Thus, the algorithm does not need to check the whole neighborhood of $u$ to determine its coloring. Any remaining $u$ will be adjacent to a 0.01-fraction of vertices in $T_i$ for all but one $i \in [k]$. Thus, for remaining $u$, a coupon collector argument shows that, by sampling an average of $O(k \log k)$ vertices, the algorithm is likely to find adjacencies between $u$ and vertices in all but one color class.
    \end{enumerate}

    Therefore, with probability $1 - \exp(-\Omega(n/k))$, every vertex is either strongly inferred with respect to the core or strongly inferred with respect to the vertices whose coloring is strongly inferred by the core.
    Vertices in the former case will be colored by the first call of \textsc{Local}-\GoodCoreAdj. 
    Vertices in the latter case will be colored by sampling vertices and testing if they are strongly inferred by the color classes defined by the core in \BuildGoodCore. 
    Therefore, in the case that a unique coloring exists and the good event of \Cref{thm:good_property_probability} holds, our LCA will provide local access to it on every vertex; no vertices output \textsc{FAIL}. 

    \paragraph{Average-case probe complexity and time.}

    First, by \Cref{thm:good_property_probability}, with probability $1 - \exp(-\Omega(n/k))$, at least $9/10$ of subgraphs of size $\nu$ are good cores. Thus, the average number of subgraphs of size $\nu$ sampled is at most 
    $O(1) + \exp(-\Omega(n/k)) \cdot \exp(o(n/k)) = O(1).$ This corresponds to $O(\nu)$ probes.

    Next, also by \Cref{thm:good_property_probability}, with probability $1 - \exp(-\Omega(n/k))$, every vertex $u$ not in the core has the property that, in each color class, a constant 
    fraction of the vertices are uniquely colorable via the core and adjacent to $u$. 
    Therefore, since for $u$ we must find a neighbor per color class that is also adjacent to the core 
    (which adds an $O(k \log k)$ factor by a coupon-collector argument), 
    the average number of vertices sampled to find the coloring of a second-level colorable vertex is $O(k \log k) + \exp(-\Omega(n/k)) \cdot n = O(k \log k).$ 

    The maximum expected probe complexity per vertex (over the randomness $\vec{r}$) is, therefore, $O(\nu + k \log k) = O(k^9 \log k)$. The average-case time complexity is $O(\nu^3/k) = O(k^{26} \text{polylog}(k))$.
\end{proof}

\end{document}